\documentclass[article]{siamart0516}

\usepackage{verbatim}
\usepackage{cite}
\usepackage{lipsum}
\usepackage{amsfonts}
\usepackage{graphicx}
\usepackage{epstopdf}
\usepackage{algorithmic}
\usepackage{amsopn}
\usepackage[table]{colortbl}
\usepackage{tikz}
\usepackage[ntheorem]{empheq}
\usetikzlibrary{arrows, positioning}
\tikzset{
    >=stealth',
    user/.style={
           rectangle,
           rounded corners,
           draw,
           minimum height=2em,
           minimum width=1cm,
           text centered},
    relay/.style={
           rectangle,
           rounded corners,
           draw,
           minimum height=1em,
           minimum width=1cm,
           text centered},
    pil/.style={
           ->,
           shorten <=2pt,
           shorten >=2pt},
    pil_rev/.style={
           <-, dashed,
           shorten <=2pt,
           shorten >=2pt},
}
\ifpdf
  \DeclareGraphicsExtensions{.eps,.pdf,.png,.jpg}
\else
  \DeclareGraphicsExtensions{.eps}
\fi

\newtheorem{remark}{Remark}

\newcommand{\Q}{\mathbb{Q}}

\newcommand{\Z}{\mathbb{Z}}
\newcommand{\R}{\mathbb{R}}
\newcommand{\C}{\mathbb{C}}
\newcommand{\F}{\mathbb{F}}

\newcommand{\mb}[1]{\mathbf{#1}}
\newcommand{\tb}[1]{\textbf{#1}}
\newcommand{\mc}[1]{\mathcal{#1}}
\newcommand{\mf}[1]{\mathfrak{#1}}

\DeclareMathOperator{\vecspan}{span}
\DeclareMathOperator{\colspan}{colspan}
\DeclareMathOperator{\gl}{GL}
\DeclareMathOperator{\mat}{Mat}
\DeclareMathOperator{\rk}{rk}

\DeclareMathOperator{\diag}{diag}

\DeclareMathOperator*{\argmin}{arg\,min}
\DeclareMathOperator*{\argmax}{arg\,max}
\DeclareMathOperator{\snr}{SNR}
\DeclareMathOperator{\GL}{GL}
\DeclareMathOperator{\vol}{vol}
\DeclareMathOperator{\lip}{Lip}

\DeclareMathOperator{\spn}{span}
\def\GL{\operatorname{GL}}
\def\SL{\operatorname{SL}}

\def\thapp{\Theta_\Lambda^{\mathfrak{A}}}
\title{{An Approximation of Theta Functions with Applications to Communications}\thanks{Submitted to the editors \today.
\funding{This work was supported by Academy of Finland grants 276031, 282938, and 303819. \newline 
A.~Barreal was with the Department of Mathematics and Systems Analysis, Aalto University, Finland, when this work was carried out, and is currently with Coop IT Digital Analytics and AI, Switzerland. T.~Damir, R.~Freij-Hollanti, and C.~Hollanti are with the Department of Mathematics and Systems Analysis, Aalto University, Finland.}}}

\author{
	Amaro Barreal\thanks{Email: \email{am.barreal@gmail.com}}
  \and Mohamed Taoufiq Damir\thanks{Email: \email{mohamed.damir@aalto.fi}} \and Ragnar Freij-Hollanti\thanks{Email: \email{ragnar.freij@aalto.fi}} \and Camilla Hollanti\thanks{Email: \email{camilla.hollanti@aalto.fi}}
}

\begin{document}

\maketitle

% REQUIRED
%\cami{I have made commands for commenting, see above the title in the tex. Please use comment colors for ANY edits or comments to keep track of changes!}

\begin{abstract}
Computing the theta series of an arbitrary lattice, and more specifically a related quantity known as the flatness factor, has been recently shown to be important for lattice code design in various wireless communication setups. However, the theta series is in general not known in closed form, excluding a small set of very special lattices. In this article, motivated by the practical applications as well as the mathematical problem itself, a simple approximation of the theta series of a lattice is derived. A rigorous analysis of its accuracy is provided.  

In relation to this, maximum-likelihood decoding in the context of compute-and-forward relaying is studied. Following previous work, it is shown that the related metric can exhibit a flat behavior, which can be characterized by the flatness factor of the decoding function. Contrary to common belief, we note that the decoding metric can be rewritten as a sum over a random lattice only when at most two sources are considered. Using a particular matrix decomposition, a link between the random lattice and the code lattice employed at the transmitter is established, which leads to an explicit criterion for code design, in contrast to implicit criteria derived previously. Finally, candidate lattices are examined with respect to the proposed criterion using the derived theta series approximation. 
\end{abstract}

% REQUIRED

% REQUIRED

\section{Introduction}
\label{sec:introduction}
Lattices are mathematical objects which have become indispensable for code design in many areas of wireless communications, as many design criteria for reliable performance rely on the discrete and algebraic structure of lattices. Despite their deceptively simple structure, many computational problems related to lattices are extremely challenging, such as the famous \emph{shortest vector problem} or related \emph{closest vector problem}. In particular, as the same lattice can be generated by distinct bases, a natural problem is to find a basis consisting of shortest vectors, a problem so hard that cryptographic protocols have been developed around it. Moreover, even enumerating vectors of certain lengths is very difficult. The generating function for the number of elements in a lattice of a given norm is known as the \emph{theta series} of the lattice. This is an interesting object in its own right, and it is not surprising that it is only known in closed form for a very small set of highly structured lattices. 

From an applications perspective, it has been recently shown that code design in various areas of wireless communications  and cryptography can profit from studying the theta series of certain involved lattices, \emph{e.g.}, for \emph{wiretap code design} \cite{ling, campello},  \emph{dither avoidance in lattice noise quantization} \cite{linglu}, or \emph{compute-and-forward relaying} \cite{nazer}. Compute-and-forward relaying is a promising physical layer network coding protocol proposed in the award-winning paper \cite{nazer}, and exploits the natural effects of interference by decoding linear combinations of the transmitted messages at the intermediate relays to achieve high computation rates. It will be the main applicational focus in this paper. For more details, see Sec. \ref{sec:caf}.

Originally, a relay operating under the compute-and-forward protocol would first scale the received signal before applying a minimum-distance decoder to obtain an estimate of the desired linear combination of the codewords. The decoding error probability for this decoding procedure was studied in \cite{feng}. It was later in \cite{belfiore} where \emph{maximum-likelihood} (ML) decoding at the relay was first considered. An approach to lattice code design for compute-and-forward was simultaneously derived therein, as well as in \cite{belfiore2}, and the first efficient decoding algorithm was proposed in dimension~$1$. The subsequent work \cite{mejri} builds upon those innovative articles and continues the investigation towards efficient decoding algorithms, an example of which is derived for Gaussian channels without fading. The fundamental work carried out in \cite{belfiore,belfiore2} is essential for code design considerations, as it introduces the notion of the \emph{flatness factor} of a lattice and utilizes it to derive an implicit lattice code design criterion. This criterion is indirect in the sense that it relates to an uncontrollable sum of random lattices and not to the code lattices themselves, where the randomness is enabled by the physical channel. It is also noteworthy that following the work \cite{belfiore}, the common belief has been that this sum can be rewritten as a sum over elements of a lattice for any number of transmitters. This is, as shown in this article, only the case if at most two sources are considered, the case studied empirically in \cite{belfiore,belfiore2}. More recently, the compute-and-forward protocol has been extended to  more general rings of algebraic integers \cite{cong-ring}.

The article is structured as follows. We start by recalling the most important results related to lattices in Section~\ref{sec:lattices}. The concepts of theta series and flatness factor are subsequently introduced in Section~\ref{sec:theta}, wherein we derive a simple but accurate approximation of the theta series and, consequently, the flatness factor of a lattice (cf. Def. \ref{theta-flatness} and the equations beneath). We provide a rigorous study on the accuracy of the approximation, alongside with some illustrating plots and discussion. In Section~\ref{sec:caf}, we summarize the compute-and-forward protocol and, following \cite{belfiore,belfiore2}, investigate the behavior of the ML decoding metric in terms of its flatness factor. Adopting certain restrictions, we establish a link between the resulting random lattice and the code lattice, allowing for an explicit lattice code design criterion. Namely, we show that in order to maximize the flatness factor of the random lattice, it suffices to maximize that of the code lattice. We then make use of the derived theta series approximation to investigate different lattices in varying dimensions with respect to the design criterion. The main contributions are the following. 
\begin{itemize}
	\item In Theorem~\ref{thm:theta_approx} we derive a simple but accurate approximation of the theta series of a lattice. For a fixed dimension, the approximation is merely a rational function, and in most cases significantly outperforms a mere series truncation approximation.
		Such an easy-to-compute approximation is important, as approximating the theta series is crucial in many lattice related applications as closed form expressions are unknown even for most deterministic lattices. In particular, the approximation also yields an approximation of the lattice flatness factor via Def. \ref{theta-flatness}, which relates to, \emph{e.g.}, compute-and-forward decoding, wiretap coset code design, smoothing parameter in cryptography, and dither avoidance in lattice noise quantization as mentioned above. 
		\item We provide a rigorous analysis on the accuracy of the approximation as well as some intuition and discussion on its qualities (Sec. 3.1 and 3.2.). More precisely:
\begin{itemize}
\item We show that our approximation is, on average (over the space of all lattices), below the value of the complete theta series. Furthermore, we show that the error term, on average, goes to zero both when $q\rightarrow 1$ (resp. $\sigma\rightarrow \infty$) and when $q\rightarrow 0$ (resp. $\sigma\rightarrow 0$).
\item We show that for any fixed $j$, there is a threshold $\sigma_j$ such that our approximation is larger than $\Theta_j$ (the $j$-th term truncation) for any $\sigma>\sigma_j$.
\item We show that we are better than the first term truncation, except possibly for some small values of $\sigma$ for some lattices, depending on the kissing number of the lattice. 
\end{itemize}
    Combining these results as well as our numerical examples, we are convinced that there is a very strong basis for using this approximation. 

\item We provide a simple explicit formula for computing the approximation for  even dimensions. An explicit formula can be also derived for odd dimensions.
\item We provide an alternative description of the error term that now more explicitly depends on the first minimum and not on the Lipschitz constant. 
\item As a special case, we motivate the accuracy by a heuristic on the minimality of the error term when the lattice is chosen to be well-rounded of dimension 2 or 3. This case is of particular interest for wireless communications. 
	
	\item The compute-and-forward ML decoding framework  is slightly generalized in Proposition~\ref{prop:decoding_manipulation} to allow for more general lattices than in previous work. While the analysis of the function can become more difficult depending on the matrix decomposition used, the decoding procedure can nonetheless be executed by the relay also in this more general setting.
	
	\item In Lemma~\ref{lem:sum_of_lattices}, we note that the decoding metric can be rewritten as a sum over elements of a lattice only for two sources, rectifying the common belief that this holds for any number of transmitting sources. 
	
	\item Theorem~\ref{thm:lattice_equivalent} establishes a link between the code lattice and the random lattice involved in the ML-decoding metric, allowing to state an explicit design criterion for the code lattice, in contrast to previous implicit criteria. 
	
	\item Finally, various lattices are examined using the explicit design criterion and derived theta series approximation.
\end{itemize}

\section{Lattices}
\label{sec:lattices}

This section is dedicated to acquainting the reader with basic concepts in lattice theory. In this article, a vector is labeled in bold, $\mb{v}$, and is always represented as a column vector. 
\begin{definition}
	A \emph{lattice} $\Lambda \subset \R^n$ is a discrete\footnote{By discrete we mean that the metric on $\R^n$ defines the discrete topology on $\Lambda$.} subgroup of $\R^n$ with the property that there exists a basis $\left\{\mb{b}_1,\ldots,\mb{b}_t\right\}$ of $\R^n$ such that 
	\begin{align}
		\Lambda = \bigoplus\limits_{i=1}^{t}{\mb{b}_i \Z}.
	\end{align} 
We say that $\left\{\mb{b}_1,\ldots,\mb{b}_t\right\}$ is a $\Z$-basis of $\Lambda$, thus $\Lambda \cong \Z^t$ as abelian groups. We call $t = \rk(\Lambda) \le n$ the \emph{rank}, and $n$ the \emph{dimension} of $\Lambda$.

A lattice $\Lambda' \subset \R^n$ such that $\Lambda' \subset \Lambda$ is called a \emph{sublattice}\footnote{If $\dim(\Lambda) = \dim(\Lambda')$, then the index $\left|\Lambda/\Lambda'\right|$ is finite.} of $\Lambda$. 
\end{definition}

More conveniently, we can define a \emph{generator matrix} $M_{\Lambda} := \begin{bmatrix} \mb{b}_1 & \cdots & \mb{b}_t \end{bmatrix}$, so that every point $\mb{x} \in \Lambda$ can be expressed as $\mb{x} = M\mb{z}$ for some vector $\mb{z} \in \Z^t$. Henceforth we will only consider \emph{full} lattices, that is, where $t = n$.

\begin{remark}
	Given a pair of full lattices $\Lambda_1 \subseteq \Lambda_2$, we will say that $\Lambda_1$ is \emph{nested} in $\Lambda_2$. We refer to $\Lambda_2$ as the \emph{fine} lattice, and to $\Lambda_1$ as the \emph{coarse} lattice. Similarly, a sequence $\Lambda_1,\ldots,\Lambda_s$ of lattices is \emph{nested} if $\Lambda_1 \subseteq \Lambda_2 \subseteq \cdots \subseteq \Lambda_s$.  
\end{remark}

Given a full lattice $\Lambda \subset \R^n$, the $i^{\mathrm{th}}$ \emph{successive minimum} of $\Lambda$, for $i = 1,\ldots,n$, is defined as
	\begin{align}
		\lambda_i = \lambda_{i}(\Lambda) := \left(\inf \left\{\left. r \right| \dim(\spn(\Lambda \cap \mathcal{B}_{\mathbf{0}}(r)))\geq i \right\}\right)^2,
	\end{align}
	 where $\mathcal{B}_{\mathbf{0}}(r)$ is the sphere of radius $r$ centered at the origin. 
	 The first minimum, $\lambda_1 = \min\limits_{\mb{x}\in\Lambda} ||\mb{x}||^2$ is referred to as the (square) \emph{minimal norm} of $\Lambda$, which exists due to the discreteness property of the lattice. If all successive minima are equal, $\lambda_1 = \cdots = \lambda_n$, the lattice is called \emph{well-rounded}.

Consider now a lattice $\Lambda$ with generator matrix $M_{\Lambda} = \left[\mb{b}_i\right]_{1 \le i \le n}$. The \emph{fundamental parallelotope} of $\Lambda$ is defined as
\begin{align}
	\mc{P}_{\Lambda} := \left\{\left.\sum\limits_{i=1}^{n}{\mb{b}_i z_i} \right| 0 \le z_i < 1 \right\},
\end{align}
and we define the \emph{volume} of $\Lambda$ to be the volume of $\mc{P}_{\Lambda}$, 
	\begin{align}
		\vol{\Lambda} := \vol{\mc{P}_{\Lambda}} = \left|\det(M_{\Lambda})\right|.
	\end{align}
Note that $\vol{\Lambda}$ is independent of the choice of the generator matrix $M_{\Lambda}$. We can easily compute the volume of a sublattice $\Lambda' \subset \Lambda$ as 
$\vol{\Lambda'} = \vol{\Lambda}\left|\Lambda/\Lambda'\right|$.

A further useful function, not only for coding-theoretic purposes, is a \emph{lattice quantizer} $Q_{\Lambda}$, a function that maps every point $\mb{y} \in \R^n$ to its closest point in the lattice. This function allows us to define a \emph{modulo-lattice} operation, $\mb{y}\ (\bmod\ \Lambda) := \mb{y}-Q_{\Lambda}(\mb{y})$. Given a lattice $\Lambda$ and a lattice quantizer $Q_{\Lambda}$, we can associate to each lattice point $\mb{x} \in \Lambda$ its \emph{Voronoi cell}, the set
\begin{align}
	\mc{V}_{\Lambda}(\mb{x}) := \left\{\left.\mb{y} \in \R^n \right| Q_{\Lambda}(\mb{y}) = \mb{x}\right\}.
\end{align}
The Voronoi cell around the origin, $\mc{V}(\Lambda) := \mc{V}_{\Lambda}(\mb{0})$, is called the \emph{basic Voronoi cell} of $\Lambda$. 

With the above definitions, we can now define the notion of a \emph{nested lattice code}, an object widely used for code construction in different communications scenarios. 

\begin{definition}
	Let  $\Lambda_C \subset \Lambda_F$ be a pair of nested lattices. We define a \emph{nested lattice code} $\mc{C}(\Lambda_C,\Lambda_F)$ as the set of representatives
	\begin{equation}
		\mc{C}(\Lambda_C,\Lambda_F) := \left\{\left.\left[\mb{x}\right] \in \Lambda_F\ (\bmod\ \Lambda_C) \right| \mb{x} \in \Lambda_F \right\} = \Lambda_F \cap \mc{V}(\Lambda_C).	
	\end{equation}
	
	The \emph{code rate} of $\mc{C}(\Lambda_C,\Lambda_F)$ in bits per dimension is
	\begin{equation}
		\mc{R} = \frac{1}{n}\log_2{|\mc{C}(\Lambda_C,\Lambda_F)|} = \frac{1}{n}\log_2{\frac{\vol{\Lambda_C}}{\vol{\Lambda_F}}} = \frac{1}{n}\log_2{|\Lambda_F/\Lambda_C|}.
	\end{equation}
\end{definition}

%\begin{comment}
Note that some coset representatives fall on the boundary of $\mc{V}(\Lambda_C)$, and need to be selected systematically. We illustrate the introduced concepts in Figure~\ref{fig:nested_code} below.
\begin{figure*}[!t]
\centering
\begin{minipage}[t]{0.48\textwidth}
	\includegraphics[trim={10cm 0 9cm 0},clip,scale=0.30]{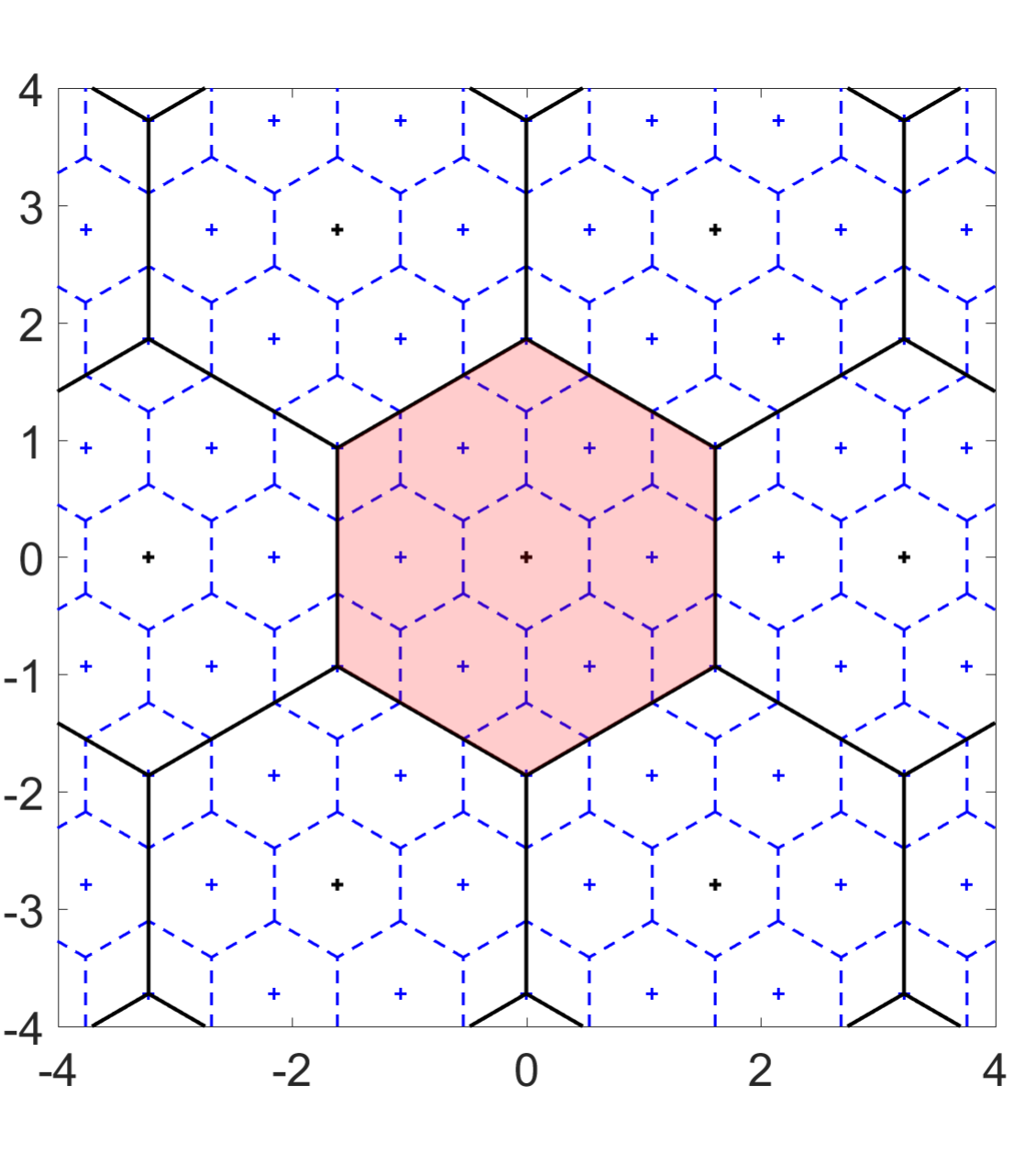}
\end{minipage}
\begin{minipage}[t]{0.48\textwidth}
	\includegraphics[trim={10cm 0 9cm 0},clip,scale=0.30]{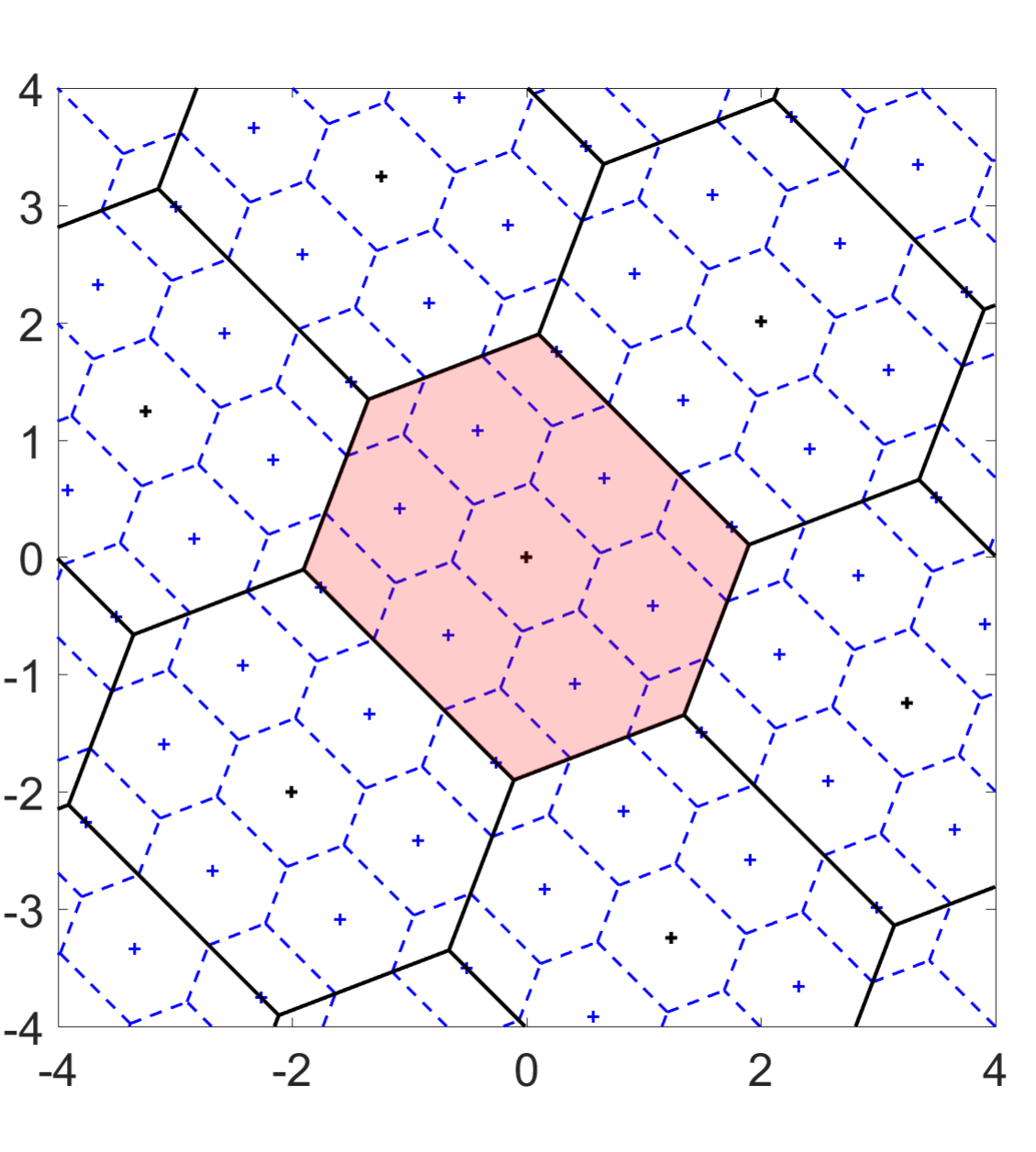}
\end{minipage}
\caption{Nested lattices $\Lambda_C \subset \Lambda_F = 3\Lambda_C$ with the Voronoi cells around each lattice point of the coarse (solid) and fine (dashed) lattices. On the left figure we fix $\Lambda_C = A_2$, the hexagonal lattice, and on the right figure $\Lambda_C = \Psi\left(\mathcal{O}_{\Q(\sqrt{5})}\right)$, the lattice obtained via the canonical embedding $\Psi$ of the ring of integers of the algebraic number field $\Q(\sqrt{5})$. 
The centered Voronoi cell $\mc{V}(\Lambda_C)$ (red) contains a set of representatives for a nested lattice code $\mc{C}(\Lambda_C,\Lambda_F)$ of cardinality $|\mc{C}(\Lambda_C,\Lambda_F)| = \left|\Lambda_F/\Lambda_C\right| = 9$.}
\label{fig:nested_code}
\end{figure*}
%\end{comment}

\section{The Theta Series and Flatness Factor of a Lattice}
\label{sec:theta}

In this section, we introduce the objects of main interest for this article: the \emph{theta series}, and a related quantity, the \emph{flatness factor} of a lattice. 

\begin{definition}
\label{def:theta}
	Let $\Lambda \subset \R^n$ be a full lattice. For each $r \in \R$, define
	\begin{align}
	\Omega_{\Lambda}(r) &:= \left|\left\{\mb{x}\in \Lambda \mid \left|\left| \mb{x}\right|\right|^2 = r \right\}\right| \\
	\Sigma_{\Lambda}(r) &:= \left|\left\{\left.\mb{x} \in \Lambda \right| ||\mb{x}||^2 \le r \right\}\right| = \sum\limits_{0 < i \le r}{\Omega_{\Lambda}(i)}.
	\end{align}
	The \emph{theta series} of $\Lambda$ is the generating function
	\begin{align}
		\Theta_{\Lambda}(q) := 1 + \sum\limits_{r > 0}{\Omega_{\Lambda}(r)q^r} = \sum\limits_{\mb{x}\in \Lambda}{q^{\left|\left| \mb{x}\right|\right|^2}}.
	\end{align}
\end{definition}

\begin{remark}
	The theta series converges absolutely if $0 \le q < 1$. We further note that 
\begin{equation}
	\argmin\limits_{r>0}\left\{\Omega_{\Lambda}(r) > 0 \right\} = \lambda_1, \qquad
	\min\limits_{r>0}\left\{\Omega_{\Lambda}(r)>0\right\} = \kappa(\Lambda),
\end{equation} 
where $\kappa(\Lambda)$ is the kissing number of $\Lambda$. It is thus clear that $\Theta_{\Lambda}(q)$ encodes important features of $\Lambda$.

More generally, the theta series is defined in terms of a complex variable $q = e^{\pi i z}$, where $z \in \C$. In this case, $\Theta_{\Lambda}(q)$ is a holomorphic function for $\Im(z) \ge 0$. For the purposes of this article, however, it suffices to view $\Theta_{\Lambda}(q)$ as a formal power series in a real variable $q$. 
\end{remark}

Although of great importance, the theta series is unfortunately only known in closed form for a handful of lattices, for example those tabulated in Table~\ref{tab:lattices} below, and is usually given in terms of the \emph{Jacobi theta functions} 
\begin{equation}
	\theta_2(q) = \sum\limits_{k=-\infty}^{\infty}{q^{\left(k+\frac{1}{2}\right)^2}},\quad
	\theta_3(q) = \sum\limits_{k=-\infty}^{\infty}{q^{k^2}}, \quad
	\theta_4(q) = \sum\limits_{k=-\infty}^{\infty}{(-q)^{k^2}}.
\end{equation}

\begin{table*}[t!]
\begin{center}
\begin{small}
	\begin{tabular}{|c||c|c|c|l|}
		\hline
		Lattice & Dim & $\lambda_1$ & $\vol{\Lambda}$ & $\Theta_{\Lambda}(q)$ \\
		\hline \hline
		\rule{0pt}{.6cm}
		\shortstack{$\Z^n$\\ \tiny{Integer}} & $n \ge 1$ & 1 & 1 & $\theta_3^n(q)$ \\
		\hline
		\rule{0pt}{.6cm}
		\shortstack{$D_n$\\ \tiny{Checkerboard}} & $n \ge 3$ & 2 & 2 & $\frac{1}{2}(\theta_3^n(q)+\theta_4^n(q))$ \\
		\hline
		\rule{0pt}{.6cm}
		\shortstack{$A_2$\\ \tiny{Hexagonal}} & 2 & 1 & $\sqrt{\frac{3}{4}}$ & $\theta_2(q)\theta_2(q^3)+\theta_3(q)\theta_3(q^3)$ \\
		\hline
		\rule{0pt}{.6cm}
		\shortstack{$E_8$\\ \tiny{Gosset}} & 8 & 2 & 1 & $\frac{1}{2}(\theta_2^8(q)+\theta_3^8(q)+\theta_4^8(q))$ \\
		\hline
		\rule{0pt}{.6cm}
		\shortstack{$K_{12}$\\ \tiny{Coxeter-Todd}} & 12 & 4 & 27 & \vtop{\hbox{\strut $\frac{9}{32}\theta_2^6(q)\theta_2^6(q^3)+\left(\theta_2(q^4)\theta_2(q^{12})+\theta_3(q^4)\theta_3(q^{12})\right)^6$}\hbox{\strut $+ \frac{45}{16}\theta_2^4(q)\theta_2^4(q^3)\left(\theta_2(q^4)\theta_2(q^{12})+\theta_3(q^4)\theta_3(q^{12})\right)^2$}} \\
		\hline
		\rule{0pt}{.6cm}
		\shortstack{$L_{24}$\\ \tiny{Leech}} & 24 & 4 & 1 & $\frac{1}{2}(\theta_2^8(q)+\theta_3^8(q)+\theta_4^8(q))^3-\frac{45}{16}(\theta_2(q)\theta_3(q)\theta_4(q))^8$ \\
		\hline 
	\end{tabular}
	\end{small}
	\caption{Various important lattices and their basic attributes.}	
	\label{tab:lattices}
\end{center}
\end{table*}

Even so, the Jacobi theta functions are by no means simple functions, but rather hard. The reason for this small set of lattices with known closed form theta series is that efficient counting of lattice points in domains in arbitrary dimensions is still an open problem. While many results have been obtained over the last two decades, such as the results in \cite{widmer,henk,fukshansky}, the settings are so general that the upper bounds on the number of lattice points in bounded domains are far from being tight, even for very simple lattices and domains. Thus, being able to efficiently compute even an approximated version of the theta series of an arbitrary lattice is a problem which is interesting in its own right. 

As additional motivation, and as we shall see in later parts of this article, recent work on lattice code design in different wireless communication scenarios \cite{ling, linglu, belfiore, belfiore2} has led to considering the \emph{flatness factor} of a lattice, which itself is directly related to the theta series of the lattice -- see Def. \ref{theta-flatness} and the equations beneath in Section \ref{subsec:flatnessfactor}. 

We define the \textit{gamma function} and the \textit{incomplete gamma function}  for $a\in \mathbb{R}, x>0$ respectively as
$$\Gamma(a):= \int_{0}^{\infty}t^{a-1}e^{-t}dt,\quad \Gamma(a,x):= \int_{x}^{\infty}t^{a-1}e^{-t}dt.$$
For an integer argument $a=n\in\mathbb{N}$, we have $\Gamma(n)=(n-1)!$\,.

\begin{theorem}
\label{thm:theta_approx}
	Let $\Lambda \subset \R^n$ be a full lattice with volume $\vol{\Lambda}$ and minimal norm $\lambda_1$. The theta series $\Theta_{\Lambda}(q)$, where $0 \le q < 1$, can be expressed as 
	\begin{align}
		\Theta_{\Lambda}(q) &= (1-q^{\lambda_1}) - \frac{\log(q)\lambda_1^{\frac{n}{2}+1}\pi^{\frac{n}{2}}}{\Gamma\left(\frac{n}{2}+1\right)\vol{\Lambda}}\int\limits_{1}^{\infty}{t^{\frac{n}{2}}q^{\lambda_1 t} dt} + \Xi(\Lambda,n,L,q),
	\end{align}
where 
	\begin{align}
		\Xi(\Lambda,n,L,q) = -C(\Lambda,n,L)\log(q) \lambda_1 \int\limits_{1}^{\infty}{t^{\frac{n-1}{2}}q^{\lambda_1 t} dt}.
	\end{align}
	The constant $C(n,\Lambda,L)$ depends on $n$, $\Lambda$, and a Lipschitz constant $L$.
\end{theorem}

We will build up the proof using a series of propositions.  

\begin{proposition}
\label{prop:theta1}
	Let $\Lambda \subset \R^n$ be a full lattice with minimal norm $\lambda_1$. Then, 
	\begin{align}
		\Theta_{\Lambda}(q) = \left(1-q^{\lambda_1}\right) - \log(q)\lambda_1\int\limits_{1}^{\infty}{\Sigma_{\Lambda}(\lambda_1 t)q^{\lambda_1 t} dt}.
	\end{align}
\end{proposition}

\begin{proof}
	Using the elementary fact $\int_{a}^{\infty}{q^t dt} = -\frac{q^a}{\log(q)}$ for $a \ge 0$, we write  
	\begin{align}
		\Theta_{\Lambda}(q) &= \sum\limits_{\mb{x} \in \Lambda}{q^{||\mb{x}||^2}} = \sum\limits_{\mb{x} \in \Lambda}{\int\limits_{||\mb{x}||^2}^{\infty}{-\log(q) q^t dt}} \\
		&= -\int\limits_{0}^{\infty}\left|\left\{\left.\mb{x} \in \Lambda \right| ||\mb{x}||^2 \le t \right\}\right| \log(q) q^t dt \\
		&= -\int\limits_{0}^{\infty}{\Sigma_{\Lambda}(t) \log(q) q^t dt}. 
	\end{align}
	
	We observe that $\Sigma_{\Lambda}(\lambda_1 t) \equiv 1$ for $t \in \left[0,1\right)$, thus by substituting $t \mapsto \lambda_1 t$ and splitting the integration range, we have
	\begin{align}
		\Theta_{\Lambda}(q) &= -\int\limits_{0}^{1}{\Sigma_{\Lambda}(\lambda_1 t) \log(q) \lambda_1 q^{\lambda_1 t} dt} -\int\limits_{1}^{\infty}{\Sigma_{\Lambda}(\lambda_1 t) \log(q) \lambda_1 q^{\lambda_1 t} dt} \\
		&= \left(1-q^{\lambda_1}\right)- \log(q)\lambda_1\int\limits_{1}^{\infty}{\Sigma_{\Lambda}(\lambda_1 t) q^{\lambda_1 t} dt}.
	\end{align}
\end{proof}

The next step is to estimate the quantity $\Sigma_{\Lambda}(r)$, which counts the number of lattice points in an $n$-sphere of radius $\sqrt{r}$. To that end, we first need the following technical definition and a related lemma. 

\begin{definition}
	Let $S \subset \R^n$ be a bounded convex set. We say that $S$ is \emph{$(n-1)$-Lipschitz parametrizable}, and write $S \in \lip(n,T,L)$, if there are $T$ maps $\phi_1,\ldots,\phi_T: \left[0,1\right]^{n-1} \to S$, the union of images of which cover $S$, and satisfying for all $1 \le i \le T$ the Lipschitz condition 
	\begin{align}
		\left|\phi_i(\mb{x})-\phi_i(\mb{y})\right| \le L\left|\mb{x}-\mb{y}\right|.
	\end{align}
\end{definition}

\begin{lemma}\cite[p.~128, Thm.~2]{lang}
\label{lem:lipschitz}
	Let $D \subset \R^n$ be such that $\partial D$ is $(n-1)$-Lipschitz parametrizable, that is, $\partial D \in \lip(n,T,L)$, and let $\Lambda \subset \R^n$ be a full lattice of volume $\vol{\Lambda}$. Then, 
	\begin{align}
		\left|\left\{\left.\mb{x} \right| \mb{x} \in \Lambda \cap rD\right\} \right| = \frac{\vol{D}}{\vol{\Lambda}}r^n + O(r^{n-1}),
	\end{align}
	where the error term $O(r^{n-1})$ depends on $\Lambda$, $n$, and the Lipschitz constant $L$. 
\end{lemma}

Using the above lemma, we can now prove the next result.
\begin{proposition}
\label{prop:theta2}
	Let $\Lambda \subset \R^n$ be a full lattice with minimal norm $\lambda_1$ and volume $\vol{\Lambda}$. Let $\Sigma_{\Lambda}(r) := \left|\left\{\left. \mb{x} \in \Lambda \right| ||\mb{x}||^2 \le r \right\}\right|$, $r \in \R_{>0}$ sufficiently large. Then, 
	\begin{align}
		|\Sigma_{\Lambda}(\lambda_1 r) - \frac{(\pi\lambda_1 r)^{\frac{n}{2}}}{\Gamma\left(\frac{n}{2}+1\right) \vol{\Lambda}}| \leq C(\Lambda,n,L) r^{\frac{n-1}{2}},
	\end{align}
	for some constant $C(\Lambda,n,L)$ that depends on the lattice, dimension, and a Lipschitz constant $L$. 
\end{proposition}

\begin{proof}
	We use Lemma~\ref{lem:lipschitz} with $D_{\lambda_1} := \mathcal{B}_{\mathbf{0}}(\sqrt{\lambda_1})$, a sphere of radius $\sqrt{\lambda_1}$ centered at the origin. Since $D_{\lambda_1}$ is bounded and convex, by \cite[Thm.~2.6]{widmer} we have $\partial D_{\lambda_1} \in \lip(n,1,L)$. 
	
We can now write
	\begin{align}
		\Sigma_{\Lambda}(\lambda_1 r) &= \left|\left\{\left. \mb{x} \in \Lambda \right| ||\mb{x}||^2 \le \lambda_1 r \right\}\right| \\
		&= \left|\left\{ \mb{x} \in \left(\Lambda \cap \mathcal{B}_{\mathbf{0}}\left(\sqrt{\lambda_1 r}\right)\right)\right\}\right| \\
		&= \left|\left\{ \mb{x} \in \Lambda \cap \left(\sqrt{r} D_{\lambda_1}\right) \right\}\right|. 
	\end{align}
	
	Using the relation $\vol{D_{\lambda_1}} = \vol{\mathcal{B}_{\mathbf{0}}(\sqrt{\lambda_1})} = \frac{(\pi\lambda_1)^{\frac{n}{2}}}{\Gamma\left(\frac{n}{2}+1\right)}$, we have
	\begin{align}
		\Sigma_{\Lambda}(\lambda_1 r) = \frac{(\pi\lambda_1 r)^{\frac{n}{2}}}{\Gamma\left(\frac{n}{2}+1\right)\vol{\Lambda}} + O(r^{\frac{n-1}{2}}),
	\end{align}
	where by Lemma~\ref{lem:lipschitz}, the error term $O(r^{\frac{n-1}{2}})$  is bounded by $C(\Lambda,n,L) r^{\frac{n-1}{2}}$ for some constant $C(\Lambda,n,L)$ that depends on the lattice, dimension, and a Lipschitz constant $L$. 
\end{proof}

We can now prove Theorem~\ref{thm:theta_approx} using the above results. 

\begin{proof}[Proof of Theorem \ref{thm:theta_approx}]
	By Proposition~\ref{prop:theta1} we start by writing 
	\begin{align}
		\Theta_{\Lambda}(q) = (1-q^{\lambda_1}) - \log(q)\lambda_1 \int\limits_{1}^{\infty}{\Sigma_\Lambda(\lambda_1 t)q^{\lambda_1 t} dt}.
	\end{align}
	
	Using the estimate for $\Sigma_{\Lambda}(r)$ derived in Proposition~\ref{prop:theta2}, we can now further manipulate the expression to read
	\begin{align}
		\Theta_{\Lambda}(q)+q^{\lambda_1}-1 &= -\log(q)\lambda_1 \int\limits_{1}^{\infty}{\Sigma_{\Lambda}(\lambda_1 t) q^{\lambda_1 t} dt} \\
		&= -\log(q)\lambda_1 \int\limits_{1}^{\infty}\left(\frac{(\pi\lambda_1 t)^{\frac{n}{2}}}{\Gamma\left(\frac{n}{2}+1\right)\vol{\Lambda}} + C(\Lambda,n,L)t^{\frac{n-1}{2}}\right) q^{\lambda_1 t} dt \\
		&= -\frac{\log(q)\pi^{\frac{n}{2}}\lambda_1^{\frac{n}{2}+1}}{\Gamma\left(\frac{n}{2}+1\right)\vol{\Lambda}}\int\limits_{1}^{\infty}{t^{\frac{n}{2}}q^{\lambda_1 t} dt}  - C(\Lambda,n,L)\log(q)\lambda_1\int\limits_{1}^{\infty}{t^{\frac{n-1}{2}}q^{\lambda_1 t} dt} \\ 
		&= -\frac{\log(q)\pi^{\frac{n}{2}}\lambda_1^{\frac{n}{2}+1}}{\Gamma\left(\frac{n}{2}+1\right)\vol{\Lambda}}\int\limits_{1}^{\infty}{t^{\frac{n}{2}}q^{\lambda_1 t} dt} + \Xi(\Lambda,n,L,q).
	\end{align}
\end{proof}

We will henceforth write $\Theta_{\Lambda}^{\mf{A}}(q)$ for the approximation $\Theta_{\Lambda}(q) - \Xi(\Lambda,n,L,q)$. The following corollary will be of use later. 
\begin{corollary}
\label{cor:theta}
	Let $\sigma^2 \in \R_{>0}$, and $q(\sigma^2) := e^{-\frac{1}{2\sigma^2}}$. Then, as a function of $\sigma^2$, we have
	\begin{align}
		\Theta_{\Lambda}^{\mf{A}}(q(\sigma^2)) &= \left(1-e^{-\frac{\lambda_1}{2\sigma^2}}\right) + \frac{(\lambda_1\pi)^{\frac{n}{2}}\lambda_1}{2\sigma^2\Gamma\left(\frac{n}{2}+1\right)\vol{\Lambda}}\int\limits_{1}^{\infty}{t^{\frac{n}{2}}e^{-\frac{\lambda_1 t}{2\sigma^2}} dt}.
	\end{align}
\end{corollary}	

Let $q=q(\sigma)=e^{-1/2\sigma^2}$. An elementary change of variable $t=\frac{\lambda_1}{2\sigma^2} z$ yields

$$\Gamma\Big(\frac{n}{2}+1,x\Big)= \Big(\frac{\lambda_1}{2\sigma^2}\Big)^{\frac{n}{2}+1}\int_{\frac{2\sigma^2 x}{\lambda_1}}^{\infty}z^{\frac{n}{2}}e^{-\frac{\lambda_1}{2\sigma^2} z}dz.$$

Let $x=\frac{\lambda_1}{2\sigma^2}$. Then

\begin{equation}\label{changeofvar}
    \int_{1}^{\infty} z^{n/2}e^{-\frac{\lambda_1}{2\sigma^2} z}dz = \Big(\frac{2\sigma^2}{\lambda_1}\Big)^{\frac{n}{2}+1}\Gamma\Big(\frac{n}{2}+1,\frac{\lambda_1}{2\sigma^2}\Big).
\end{equation}

Thus, the approximation in Theorem \ref{thm:theta_approx} becomes
\begin{equation}\label{approx}
\Theta(q)=\Theta(e^{-1/2\sigma^2})=1-e^{-\lambda_1/2\sigma^2}+\frac{\Big(\sqrt{2\sigma^2\pi}\Big)^n}{\vol(\Lambda)}\frac{\Gamma\Big(n/2+1,\frac{\lambda_1}{2\sigma^2} \Big)}{\Gamma\Big(n/2+1,0\Big)} +\Xi=\thapp(q)+\Xi,   
\end{equation}
where $\Xi=\frac{\lambda_1}{2\sigma^2}C(n,\Lambda, L)\Gamma(\frac{n}{2}+\frac{1}{2},\lambda_1)$.

The following corollary provides a recursive formula for calculating the main term $\thapp(q)$ in Theorem \ref{thm:theta_approx} whenever the dimension $n$ is even.
\begin{corollary}\label{neven}
Let $q=e^{\frac{-1}{2\sigma^2}}$ and $n$ even. Then $\thapp(q(\sigma))$  in \eqref{approx} becomes
\begin{equation*}
       \thapp(q)=1+ q^{\lambda_1}\left(-1+\frac{\pi^{\frac{n}{2}}}{\vol(\Lambda)} \sum_{i=0}^{\frac{n}{2}}\frac{\lambda_1^i 2^{\frac{n}{2}-i}\sigma^{n-2i}}{i!}\right). 
\end{equation*}

\end{corollary}
\begin{proof}

\begin{align*}\thapp(q)&=1-q^{\lambda_1}+\frac{\pi^{\frac{n}{2}}}{\vol(\Lambda)}\Big(\frac{-1}{\log(q)}\Big)^{\frac{n}{2}} q^{\lambda_1}\Big(\sum_{i=0}^{\frac{n}{2}}\frac{(-\lambda_1 \log(q))^i}{i!}\Big)
\\
&=1-q^{\lambda_1}+\frac{\pi^{\frac{n}{2}}}{\vol(\Lambda)}2^{\frac{n}{2}}\sigma^n q^{\lambda_1}\sum_{i=0}^{\frac{n}{2}}\frac{\lambda_1^i }{i! 2^i\sigma^{2i}}
\\
&= 1+ q^{\lambda_1}\left(-1+\frac{\pi^{\frac{n}{2}}2^{\frac{n}{2}}\sigma^n}{\vol(\Lambda)} \sum_{i=0}^{\frac{n}{2}}\frac{\lambda_1^i }{i! 2^i\sigma^{2i}}\right)
\\
&= 1+ q^{\lambda_1}\left(-1+\frac{\pi^{\frac{n}{2}}}{\vol(\Lambda)} \sum_{i=0}^{\frac{n}{2}}\frac{\lambda_1^i 2^{\frac{n}{2}-i}\sigma^{n-2i}}{i!}\right)
.\end{align*}
\end{proof}

\subsection{Analysis for the Accuracy of $\thapp$}
From Corollary~\ref{neven}, we get the following result, showing that $\thapp$ is indeed larger than any truncation of the theta series for values of $\sigma$ sufficiently large.

\begin{proposition}
Let $\Lambda$ be a lattice of even dimension $n$, and let $j$ be a positive integer.
Then there is a threshold value $\sigma_j\geq 0$ such that $\thapp(q(\sigma))\geq \Theta_{j,\Lambda}(q(\sigma))$ for all $\sigma \geq\sigma_j$. If the lattice $\Lambda$ satisfies $(\kappa+1)\vol(\Lambda)\leq \lambda_1^{\frac{n}{2}} \vol(\mathcal{B}_{\mathbf{0}}(1))$, then $\thapp(q(\sigma))\geq \Theta_{1,\Lambda}(q(\sigma))$ for all $\sigma \geq 0$.
\end{proposition}

\begin{proof} For simplicity, we only present the proof for even $n$. For odd $n$, it goes analogously but will look messier due to the more complicated form of the gamma function. 

Let $\lambda_i$ be the $i^{\rm th}$ successive minimum norm of $\Lambda$, and let $\kappa_i$ be the number of vectors in $\Lambda$ of norm $\lambda_i$. By definition, we then have $$\Theta_{j,\Lambda}(q(\sigma))=1+\sum_{i=1}^j q^{\lambda_i} \kappa_i \leq 1+ q^{\lambda_1} \sum_{i=1}^j \kappa_i .$$
By Corollary~\ref{neven}, it thus suffices to show that \begin{equation}\label{coeffs}
-1+\frac{\pi^{\frac{n}{2}}}{\vol(\Lambda)} \sum_{i=0}^{\frac{n}{2}}\frac{\lambda_1^i 2^{\frac{n}{2}-i}\sigma^{n-2i}}{i!} \geq \sum_{i=1}^j \kappa_i\end{equation} for large enough $\sigma$. But the left hand side of \eqref{coeffs} is a continuous and strictly increasing function in $\sigma\geq 0$, and tends to infinity as $\sigma\to\infty$. As $\sum_{i=1}^j \kappa_i$ is constant, the inequality \eqref{coeffs} holds for all large enough $\sigma$. 

To prove the second part of the theorem, it is now enough to show that \eqref{coeffs} holds for $\sigma=0$, $j=0$. But when $\sigma=0$, the only non-vanishing term in the sum is when $i=\frac{n}{2}$, so \eqref{coeffs} is equivalent to
$$\frac{\pi^{\frac{n}{2}}\lambda_1^{\frac{n}{2}}}{\vol(\Lambda)\frac{n}{2}!}\geq \kappa +1.$$ Observing that $$(\kappa+1)\vol(\Lambda)\leq \lambda_1^{\frac{n}{2}} \vol(\mathcal{B}_{\mathbf{0}}(1)))=\frac{\pi^{\frac{n}{2}}}{\frac{n}{2}!},$$ the statement of the Proposition follows.
\end{proof}

It is worth noting that there is a nice geometric interpretation of the above inequality 
\begin{equation}(\kappa+1)\vol(\Lambda)\leq \lambda_1^{\frac{n}{2}}\vol(\mathcal{B}_{\mathbf{0}}(1))).
\label{eq:kappa}
\end{equation}
Namely, the right hand side of  \eqref{eq:kappa} is the volume of the ball centered around the origin with the shortest vectors of $\Lambda$ on its boundary. The left hand side of \eqref{eq:kappa} is the volume of the union of the $\kappa + 1$ Voronoi cells centered at the origin and at the shortest vectors of $\Lambda$. Depending on which of these volumes is the largest, the inequality $\thapp\geq \Theta_{1}$ holds either for all $q$, or only for large enough $q$.

To prove that the approximation $\thapp$ is indeed closer to the actual theta function $\Theta$ than the $j^{\rm th}$ truncation $\Theta_j$ for $\sigma >\sigma_j$, it would be enough to show that $\thapp(q)\leq\Theta_\Lambda (q)$ holds for all lattices $\Lambda$ and all $0\leq q<1$. While this inequality holds for all lattices for which we can do explicit calculations, we are not able to prove it in full generality. However, it holds on average in the sense of the following theorem.

\begin{theorem}
Let $\Lambda$ be a random lattice with distribution given by the Haar measure on $\SL(n,\mathbb{R})/\SL(n,\mathbb{Z})$. Then, for every $0\leq q<1$, it holds that $$\mathbb{E}[\thapp(q)] \leq \mathbb{E}[\Theta_\Lambda (q)].$$
\end{theorem}

\begin{proof}
A straightforward application of Siegel’s mean value theorem \cite{siegel1945mean} implies that for any $t>0$, we have $$\mathbb{E}(\Sigma_{\Lambda}(t))=1+\vol(\mathcal{B}_{\mathbf{0}}(t)),$$
where $\mathcal{B}_{\mathbf{0}}(1)$ is the Euclidean ball of radius $r$. For any fixed lattice $\Lambda$, we can thus write \begin{align}
    \Theta_\Lambda^{\mathfrak{A}}(q) &= (1-q^\ell)-\log q\int\limits_{\ell}^{\infty}{q^t\mathbb{E}[\Sigma_{\Lambda}(t)-1] dt}\\
    &= -\log q\int\limits_{0}^{\ell}{q^t dt} -\log q\int\limits_{\ell}^{\infty}{q^t\mathbb{E}[\Sigma_{\Lambda}(t)-1] dt}\\
    &= -\log q \int\limits_{0}^{\infty}{q^t\left(I_{t\leq\ell} + I_{t>\ell}\mathbb{E}[\Sigma_{\Lambda}(t)-1]\right) dt},
\end{align} where $\ell$ is the (deterministic) shortest norm of $\Lambda$, and $I_E$ denotes the indicator function of the event $E$. 
Observing that $$\Theta_\Lambda(q) =  -\log(q)\int\limits_{0}^{\infty}{\Sigma_{\Lambda}(t)  q^t dt},$$
we get by linearity of the expectation and by Fubini's theorem that
\begin{align*}
    \mathbb{E}[\Theta_\Lambda (q)] - \mathbb{E}[\thapp(q)] &=-\log q   \, \mathbb{E}\left[\int\limits_{0}^{\infty}{\Sigma_{\Lambda}(t)  q^t dt} - \int\limits_{0}^{\infty}{q^t\left(I_{t\leq\lambda} + I_{t>\lambda}\mathbb{E}\left[\Sigma_{\Lambda}(t)-1\right]\right) dt}\right]\\
    %& = -\log q \int\limits_{0}^{\infty}{q^t\mathbb{E}\left[\Sigma_{\Lambda}(t)- \left(I_{t\leq\lambda} + I_{t>\lambda}(\Sigma_{\Lambda}(t)-1)\right]\right] dt}\\
    & = -\log q \int\limits_{0}^{\infty}{q^t\mathbb{E}\left[\Sigma_{\Lambda}(t)- \left(I_{t\leq\lambda} + I_{t>\lambda}(\mathbb{E}(\Sigma_{\Lambda}(t))-1)\right)\right] dt}\\
    & = -\log q \int\limits_{0}^{\infty}{q^t\mathbb{E}\left[\mathbb{E}\left[\Sigma_{\Lambda}(t)\right]- \left(I_{t\leq\lambda} + I_{t>\lambda}(\mathbb{E}(\Sigma_{\Lambda}(t))-1)\right)\right] dt}\\
       & = -\log q \int\limits_{0}^{\infty}{q^t\left(\mathbb{E}\left[\mathbb{E}\left[\Sigma_{\Lambda}(t)\right](1-I_{t>\lambda})\right]- \mathbb{E}\left[I_{t\leq\lambda} - I_{t>\lambda}\right]\right) dt}\\
     &= -\log q \int\limits_{0}^{\infty}{q^t\left( \mathbb{E}\left[\Sigma_{\Lambda}(t)\right]\mathbb{P}[t\leq \lambda]-\mathbb{P}[t\leq \lambda]+ \mathbb{P}[t>\lambda]\right) dt}\\
     &= -\log q \int\limits_{0}^{\infty}{q^t (\mathbb{P}[t\leq \lambda]\mathbb{E}\left[\Sigma_{\Lambda}(t)-1\right]+ \mathbb{P}[t>\lambda] )dt},
\end{align*}
where $\lambda$ is the (random) shortest norm of $\Lambda$. The integrand is now readily seen to be a non-negative real function, wherefore we get $$\mathbb{E}[\Theta_\Lambda (q)] \geq \mathbb{E}[\thapp(q)].$$
 \hfill \ensuremath{}
\end{proof}

\subsection{Error Term Analysis for the Point Counting Function}

The proof of Theorem \ref{thm:theta_approx} relied on an estimate (Proposition \ref{prop:theta2}) of the number of lattice points in  $\mathcal{B}_{\mathbf{0}}\left(\sqrt{\lambda_1 r}\right)$.
For the sake of completeness, we sketch an alternative proof of Proposition \ref{prop:theta2} with a slightly different error term.

Let $\Lambda=M\cdot\mathbb{Z}^n$ for some $M\in\GL_n(\mathbb{R})$. Then 

\begin{align}
 \Sigma_{\Lambda}(\lambda_1 r)&=\left|M\cdot\mathbb{Z}^n\cap \mathcal{B}_{\mathbf{0}}\left(\sqrt{\lambda_1 r}\right)\right|\\
 &= \left|\mathbb{Z}^n\cap M^{-1}\mathcal{B}_{\mathbf{0}}\left(\sqrt{\lambda_1 r}\right)\right|,
\end{align}
where $ M^{-1}\mathcal{B}_{\mathbf{0}}\left(\sqrt{\lambda_1 r}\right) =\left\{M^{-1}\mb{x}~:~\mb{x}\in \mathcal{B}_{\mathbf{0}}\left(\sqrt{\lambda_1 r}\right)  \right\} $.

Consider the tiling of $\R^n$ with unit cubes centered at the points of $\Z^n$. We interpret $\Sigma_{\Lambda}(\lambda_1 r)$ as the number of unit cubes in this tiling with centers lying inside $ M^{-1}\mathcal{B}_{\mathbf{0}}\left(\sqrt{\lambda_1 r}\right)$.
Hence,

\begin{align}
   \Sigma_{\Lambda}(\lambda_1 r)&=\vol\left( M^{-1}\mathcal{B}_{\mathbf{0}}\left(\sqrt{\lambda_1 r}\right) \right)+ \mathcal{E}_{\Lambda}(\sqrt{\lambda_1 r})\\
   &= \frac{(\pi\lambda_1 r)^{\frac{n}{2}}}{\Gamma\left(\frac{n}{2}+1\right)\vol{\Lambda}}+\mathcal{E}_{\Lambda}(\sqrt{\lambda_1 r}),
\end{align}
where $\mathcal{E}_{\Lambda}(\sqrt{\lambda_1 r})$ is bounded by the volumes of the cubes that intersect the boundary $\partial M^{-1}\mathcal{B}_{\mathbf{0}}\left(\sqrt{\lambda_1 r}\right)$. 
This volume is proportional to the (Hausdorff) surface measure of $M^{-1}\mathcal{B}_{\mathbf{0}}\left(\sqrt{\lambda_1 r}\right)$.
Thus, for $r$ large enough, the dominant term in  $ \Sigma_{\Lambda}(\lambda_1 r)$ is $\vol\left(M^{-1}\mathcal{B}_{\mathbf{0}}\left(\sqrt{\lambda_1 r}\right) \right)$.

Let 
$$ \Sigma_{\Lambda}(t)=\vol\left( M^{-1}\mathcal{B}_{\mathbf{0}}\left(\sqrt{ t}\right) \right)+ \mathcal{E}_{\Lambda}(\sqrt{t}).$$
In the following, we will consider the order of magnitude of $\mathcal{E}_{\Lambda}(\sqrt{t})$ and its relation with the error term in Theorem \ref{thm:theta_approx}.

Let $C$ be a positive constant, $f$ a real valued (integrable) function and $t_0$ a positive real number such that  $$|\mathcal{E}_{\Lambda}(\sqrt{t})|<Cf(t) \textrm{ for all } t\geq t_0,$$ \textit{i.e.}, $\mathcal{E}_{\Lambda}(\sqrt{t})=O(f(t))$.

With the notation above, we re-write the error term in Theorem \ref{thm:theta_approx} in the following form.

\begin{align}\label{new-error}
       \Xi(\Lambda,n,q)&=O\left(\log(q)\lambda_1\int\limits_{1}^{\infty}{f(\lambda_1r)q^{\lambda_1 r} dr}\right)\\ 
       &=O\left(\log(q)\int\limits_{\lambda_1}^{\infty}{f(t)q^{t} dt}\right)\\
       &\leq O\left(\log(q)\int\limits_{\lambda_1}^{t_0}{f(t)q^{t} dt}\right)
   +C\log(q)\int\limits_{t_0}^{\infty}{f(t)q^{t} dt}.
\end{align}

Note that in the proof of Theorem \ref{thm:theta_approx} we implicitly assume that $t\geq \lambda_1$. Thus, $t_0\geq \lambda_1$.
Equation \eqref{new-error} shows that any improvement on the order of magnitude of $\mathcal{E}_{\Lambda}(\sqrt{\lambda_1 r})$ will necessarily imply an improved error term $ \Xi(\Lambda,n,q)$.
\begin{remark}
With this new interpretation of $\Sigma_{\Lambda}(t)$, the main term in Theorem \ref{thm:theta_approx} remains the same, but the term $\Xi(\Lambda,n,q)$ depends on $\lambda_1$ rather than the Lipschitz constant $L$.
\end{remark}

In \cite{gotz}, Götze showed that $ \mathcal{E}_{\Lambda}(\sqrt{t})=O(t^{\frac{n-2}{2}})$ for every lattice $\Lambda\subset\mathbb{R}^n$ with $n\geq 5$. This bound is tight in the sense that $ \mathcal{E}_{\Lambda}(\sqrt{t})\neq o(t^{\frac{n-2}{2}})$ for $\Lambda=\mathbb{Z}^n$. 

Assuming that $n\geq 5$, we get 
\begin{equation}\label{Gotzeerror}
   \Xi(\Lambda,n,q)\leq O\left(\log(q)\int\limits_{\lambda_1}^{t_0}{t^{\frac{n-2}{2}}q^{t} dt}\right)
   +C\log(q)\int\limits_{t_0}^{\infty}{t^{\frac{n-2}{2}}q^{t} dt}.
\end{equation}
Let $q=e^{\frac{-1}{2\sigma^2}}$. Then inequality \eqref{Gotzeerror} becomes 
\begin{equation}
     \Xi(\Lambda,n,e^{\frac{-1}{2\sigma^2}})\leq O\left(\log(q)\int\limits_{\lambda_1}^{t_0}{t^{\frac{n-2}{2}}q^{t} dt}\right)
  - \frac{C}{2\sigma^2}\int\limits_{t_0}^{\infty}t^{\frac{n-2}{2}}e^{-\frac{t}{2\sigma^2} dt}.
\end{equation}

Thus, using the same argument as in \eqref{changeofvar} we get
 \begin{equation}
      \Xi(\Lambda,n,e^{\frac{-1}{2\sigma^2}})=O\left((2\sigma^2)^{\frac{n}{2}-1}\Gamma\left(\frac{n}{2},\frac{t_0}{2\sigma^2}\right)\right).
 \end{equation}

Finally, we write the approximation in Theorem \ref{thm:theta_approx} as
\begin{equation}\label{appprox}
    \Theta\left(e^{-1/2\sigma^2}\right)- 1+e^{-\lambda_1/2\sigma^2}=\frac{\Big(\sqrt{2\sigma^2\pi}\Big)^n}{\vol(\Lambda)}\frac{\Gamma\Big(n/2+1,\frac{\lambda_1}{2\sigma^2} \Big)}{\Gamma\Big(n/2+1,0\Big)}+\Xi(\Lambda,n,e^{\frac{-1}{2\sigma^2}}).
\end{equation}

Using integration by parts, one can prove that the incomplete gamma function satisfies the following recurrence relation
\begin{equation}\label{reccurent}
    \Gamma\left(\frac{n}{2}+1,\frac{\lambda_1}{2\sigma^2}\right)=\frac{n}{2}\Gamma\left(\frac{n}{2},\frac{\lambda_1}{2\sigma^2}\right)+\left(\frac{\lambda_1}{2\sigma^2}\right)^{\frac{n}{2}} e^{-\frac{\lambda_1}{2\sigma^2}}.
\end{equation}

Assume for instance that $t_0=\lambda_1$, then using \eqref{reccurent}, the ratio of the main and error terms in \eqref{appprox} is 
\begin{align}
   \mathcal{R}&= \frac{\Big(\sqrt{2\sigma^2\pi}\Big)^n}{\vol(\Lambda)}\frac{\Gamma\Big(n/2+1,\frac{\lambda_1}{2\sigma^2} \Big)}{\Gamma\Big(n/2+1,0\Big)}\times \frac{1}{\left((2\sigma^2)^{\frac{n}{2}-1}\Gamma\left(\frac{n}{2},\frac{t_0}{2\sigma^2}\right)\right) }  \\
   &= n\sigma^2 \frac{\vol(\mathcal{B}_0(1))}{\vol(\Lambda)} + h(\lambda_1,n,\sigma),
\end{align}

where $h(\lambda_1,n,\sigma)$ is a positive constant depending on $n,\lambda_1$ and $\sigma$.

The last equality shows that under the above assumptions, the ratio $\mathcal{R}$ is greater than one whenever $\sigma>\sqrt{\frac{\vol(\Lambda)}{n \vol(\mathcal{B}_0(1))}}$.

As a last part of this section, we mention further results on the magnitude of $\mathcal{E}_{\Lambda}(\sqrt{t})$.

Let $\mathcal{L}_n$ be the set of determinant $1$ lattices with Haar measure $\mu_n$. We call a random variable sampled from $\mathcal{L}_n$ with respect to $\mu_n$ a random lattice.

Let $\delta>0$ be a small arbitrary constant. Shmidt \cite{shmidt} proved that $ \mathcal{E}_{\Lambda}(\sqrt{t})=O(t^{\frac{n}{4}+\delta})$ for almost every lattice.

It is conjectured \cite{Ivic} that $\mathcal{E}_{\Lambda}(\sqrt{t})=O(t^{\frac{n-1}{4}+\delta/2})$.
In \cite{holmin}, the author showed that the bound $\mathcal{E}_{\Lambda}(\sqrt{t})=O(t^{\frac{n-1}{4}+\delta/2})$ holds in average for lattices of dimensions $n=2,3$, where the average is taken over any compact subset $Y$ of the space of all lattices (not necessarily of volume $1$).

In the following, we will give an example of a compact subset of $\mathcal{L}_n$.
Compact subsets of $\mathcal{L}_n$ can be obtained by the so-called Mahler's compactness criterion.
\begin{theorem}[Mahler]
The set of lattices $\Lambda\in \mathcal{L}_n $ whose shortest vector is of a fixed
length $\geq r>0$ is compact.
\end{theorem}
\begin{definition}
A lattice $\Lambda\subset\mathbb{R}^n$ is \emph{well-rounded} (abbreviated WR) if $$\spn_\mathbb{R}(S(L)) = \mathbb{R}^n.$$
We denote by $\mathcal{WR}_n$ the set of well-rounded lattices in $\mathcal{L}_n$.
\end{definition}

\begin{proposition}\label{wrcompact}
    The set $\mathcal{WR}_n$ is compact in $\mathcal{L}_n$.
\end{proposition}

\begin{proof}
Let $\Lambda$ be a lattice in $\mathcal{WR}_n$, and let $v_i\in \Lambda$ such that $||v_i||=\lambda_i(\Lambda)$ for $1\leq i\leq n$. 
Taking $\Lambda'=\spn_{\Z}(v_i~|~1\leq i\leq n)$, then $\Lambda'$ is a full-rank sub-lattice of $\Lambda$. Hence, $\vol(\Lambda')=[\Lambda:\Lambda']\geq 1$. 

On the other hand $$\prod_{i=1}^n ||v_i||=\prod_{i=1}^n\lambda_i(\Lambda)=\lambda_1(\Lambda)^n.$$ 
Recalling the Hadamard's inequality $\prod_{i=1}^n ||v_i||\geq \vol(\Lambda')$, we conclude that $$\lambda_1(\Lambda)\geq 1.$$
The result follows from Mahler's compactness criterion.
\end{proof}

Combining Proposition \ref{wrcompact} with the main result in \cite{holmin}, we get the following result.
\begin{proposition}\label{errorWR}
The bound
$$
      \Xi(\Lambda,n,e^{\frac{-1}{2\sigma^2}})=O\left((2\sigma^2)^{\frac{n-5}{4}+\delta/2}\Gamma\left(\frac{n-1}{4}+\delta/2,\frac{t_0}{2\sigma^2}\right)\right)
$$
 holds on average over $\mathcal{WR}_2$ and $\mathcal{WR}_3$.
\end{proposition}

Landau \cite{landau} proved that $\mathcal{E}_{\Lambda}(\sqrt{t})\neq o(t^{\frac{n-1}{4}})$ for any lattice of dimension $n\geq 3$. Consequently, for $n\geq 3$ we have
$$ \Xi(\Lambda,n,e^{\frac{-1}{2\sigma^2}})\neq o\left((2\sigma^2)^{\frac{n-5}{4}+}\Gamma\left(\frac{n-1}{4},\frac{t_0}{2\sigma^2}\right)\right).
$$

Proposition \ref{errorWR} shows that heuristically, the error term in Theorem \ref{thm:theta_approx} is achieves the conjectured bound over the sets $\mathcal{WR}_2$ and $\mathcal{WR}_3$. 

In the next section, we analyse the accuracy of our approximation for some well-rounded lattices, \textit{i.e.}, $\mathbb{Z}^2$, $D_3$, $D_4$, $E_8$ and $K_{12}$. In fact, most of the known lattices are well-rounded. To name just a few, we mention the local maxima of the sphere packing problem, the lattices $D_n$, $A_n$ and the orthogonal lattice $\mathbb{Z}^n$.

\subsection{Empirical Study and Discussion}
\label{subsec:approx_accuracy}

We first depict the accuracy of the approximation $\Theta_{\Lambda}^{\mf{A}}(q)$ for some of the well-known lattices tabulated in Table~\ref{tab:lattices}. We choose $q = e^{-\frac{1}{2\sigma^2}}$ and interpret $\Theta_{\Lambda}^{\mf{A}}(e^{-\frac{1}{2\sigma^2}})$ as a function in the variable $\sigma^2$. The choice of this specific indeterminate $q$ will be clarified in the subsequent sections of this article. 
\begin{figure}[h!]
%\begin{subfigure}{0.48\textwidth}
	\includegraphics[width=.48\textwidth]{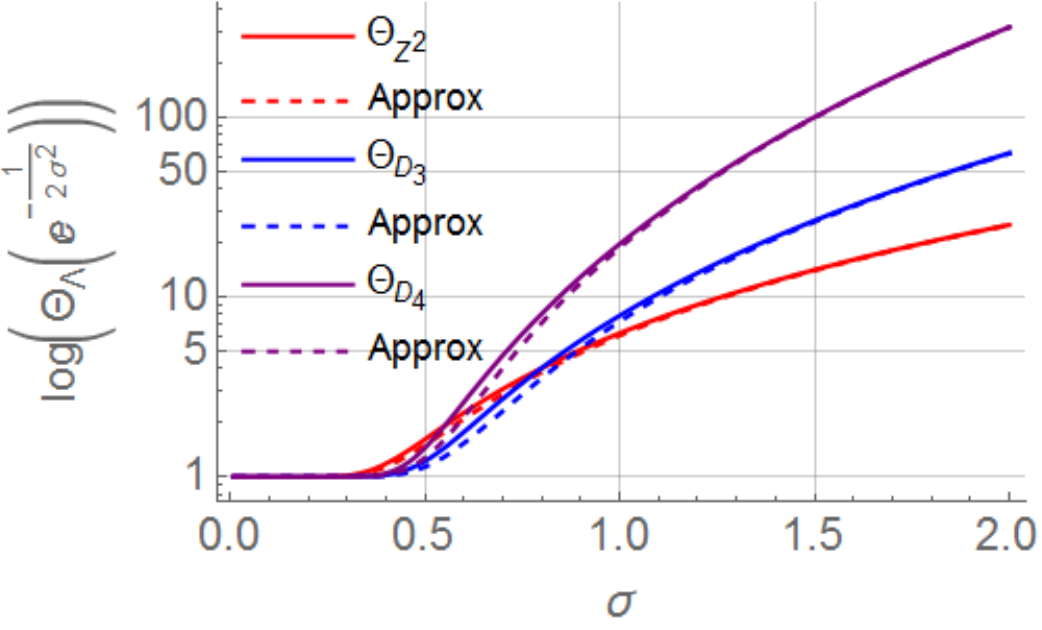}\hfill
%\end{subfigure}
%\begin{subfigure}{0.48\textwidth}
	\includegraphics[width=.48\textwidth]{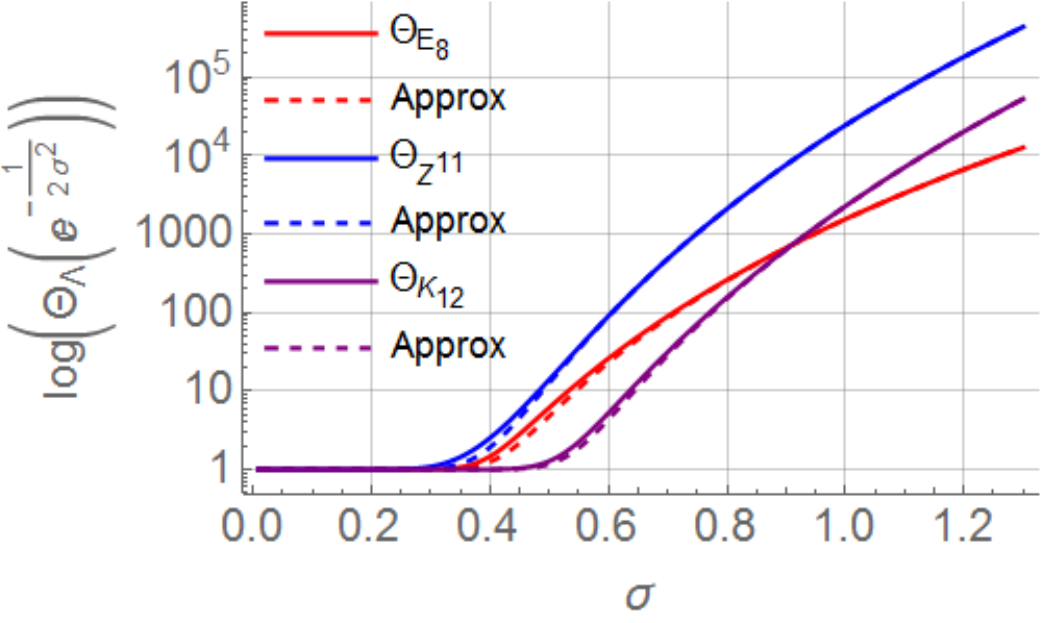}
%\end{subfigure}
\caption{Comparison of the theta function of various lattices and the derived approximation. The lhs picture depicts the theta series of low-dimensional, the rhs picture higher-dimensional lattices as a function of $\sigma^2$.}
\label{fig:theta_compare}
\end{figure}

From Figure~\ref{fig:theta_compare} it is visible that the approximation is accurate in the considered cases, even as the dimension increases. A naive way of approximating the theta series is by simply considering the first term in the power series expression, that is, $\Theta_{\Lambda}(q) \approx 1 + \kappa(\Lambda)q^{\lambda_1}$. In Figure~\ref{fig:leech_sum} below, we compare the derived approximation $\Theta_{\Lambda}^{\mf{A}}(q)$ with this truncated sum on the Leech lattice $\Lambda_{24}$. While our approximation accurately approximates the theta series $\Theta_{\Lambda_{24}}(q)$, the truncated sum very quickly diverges from the actual function, as is to be expected. 
\begin{figure}[h!]
\centering
	\includegraphics[width=.45\linewidth]{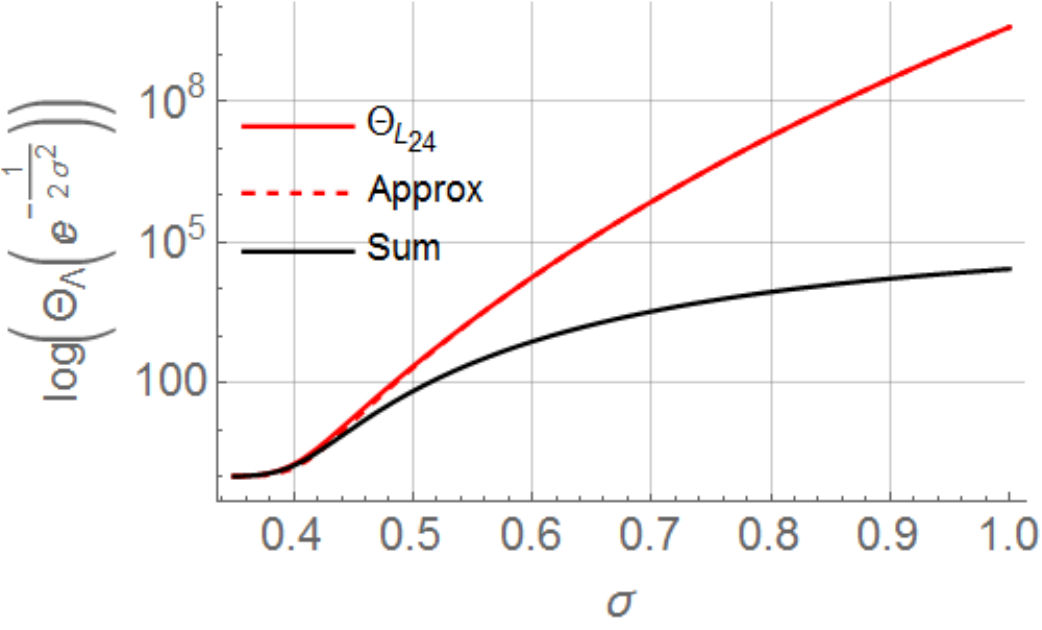}
	\caption{Comparison of $\Theta_{\Lambda}^{\mf{A}}(q)$ and a truncated sum $1+\kappa(\Lambda)q^{\lambda_1}$ of the Leech lattice $\Lambda = L_{24}$.}
	\label{fig:leech_sum}
\end{figure}

\begin{remark}
The error term in the expression from Theorem~\ref{thm:theta_approx} arises from the estimation of lattice points in an $n$-sphere, \emph{i.e.}, the estimation of $\Sigma_{\Lambda}(r)$. In its full generality, this is a hard problem. For instance, the original proof of Lemma~\ref{lem:lipschitz} in \cite{lang} is not constructive, and does not offer any insight into the involved constant. Accurately counting lattice points in more general domains is a topic of the utmost interest in lattice theory. In \cite{fukshansky}, an upper bound on the quantity $|\Lambda \cap P|$, where $\Lambda \subset \R^n$ is a full lattice, and $P \subset \R^n$ an arbitrary polytope of dimension $n' \le n$, is given. Further, \cite{widmer} gives an upper bound on $\left|\Lambda \cap S\right|$, where $S \subset \R^n$ is a bounded domain, of general \emph{narrow class} $s \ge 1$. Both mentioned results are however so general, that the upper bounds are not tight, even for low-dimensional, well-conditioned lattices.  
\end{remark}

\subsection{The Flatness Factor}
\label{subsec:flatnessfactor}

Having introduced the theta series $\Theta_{\Lambda}(q)$ of a lattice, we now define a related quantity -- the flatness factor $\varepsilon_{\Lambda}(q)$ of $\Lambda$. Consider the usual $n$-dimensional zero-mean Gaussian PDF with variance $\sigma^2$, given by
\begin{align}
	f(\mb{t},\sigma^2) = \frac{1}{(\sqrt{2\pi\sigma^2})^n} e^{-\frac{||\mb{t}||^2}{2\sigma^2}}.
\end{align}

We are interested in the case where the variable $\mb{t}$ ranges over points over a (possibly shifted) full lattice $\Lambda$, yielding for $\mb{y} \in \R^n$ the sum of Gaussian functions 
\begin{align}
	f(\Lambda+\mb{y}, \sigma^2) := \sum\limits_{\mb{x} \in \Lambda}{f(\mb{x}+\mb{y},\sigma^2)}. 
\end{align}
As a function of $\mb{y}$, $f(\Lambda+\mb{y},\sigma^2)$ is a $\Lambda$-periodic function, and defines a PDF on the basic Voronoi cell $\mc{V}(\Lambda)$ of $\Lambda$, which we refer to as the \emph{lattice Gaussian PDF}. For the centered sum $f(\Lambda,\sigma^2)$, we have the useful identity 
\begin{align}
	f(\Lambda,\sigma^2) &= \sum\limits_{\mb{x} \in \Lambda}{f(\mb{x},\sigma^2)} = \frac{1}{(\sqrt{2\pi\sigma^2})^n}\sum\limits_{\mb{x} \in \Lambda}{e^{-\frac{||\mb{x}||^2}{2\sigma^2}}} \\ 
	&= \frac{1}{(\sqrt{2\pi\sigma^2})^n}\Theta_{\Lambda}\left(e^{-\frac{1}{2\sigma^2}}\right).
\end{align}

Introduced in \cite{ling} as an information theoretic tool in the context of fading wiretap channels, the \emph{flatness factor} is a quantity which measures the deviation of the lattice Gaussian PDF from the uniform distribution on the Voronoi cell $\mc{V}(\Lambda)$. Formally, it can be defined as follows.

\begin{definition}\label{theta-flatness}
	Let $\Lambda \subset \R^n$ be a full lattice, and for $\mb{y} \in \R^n$, let $f(\Lambda+\mb{y},\sigma^2)$ denote the lattice Gaussian PDF of the lattice $\Lambda+\mb{y}$. The \emph{flatness factor} of $\Lambda$ is defined as 
	\begin{align}
		\varepsilon_{\Lambda}(\sigma^2) := \max\limits_{\mb{y} \in \R^n}\left|\frac{f(\Lambda+\mb{y},\sigma^2)}{1/\vol{\Lambda}}-1\right|.
	\end{align}	
\end{definition}

It is easy to show (see \cite{micciancio}) that the maximum of $f(\Lambda+\mb{y},\sigma^2)$ is achieved for $\mb{y} \in \Lambda$. Hence, an explicit representation of $\varepsilon_{\Lambda}(\sigma^2)$ is immediate,
	\begin{align}
		\varepsilon_{\Lambda}(\sigma^2) = \frac{\vol{\Lambda}}{(\sqrt{2\pi\sigma^2})^n}\Theta_{\Lambda}\left(e^{-\frac{1}{2\sigma^2}}\right)-1.
	\end{align}

If we define the \emph{volume-to-noise ratio}\footnote{The VNR is usually defined without the term $2\pi$ in the denominator. Here, the definition is chosen to agree with \cite{belfiore}.} (VNR) $\gamma_{\Lambda}(\sigma^2) := \frac{\vol{\Lambda}^{\frac{2}{n}}}{2\pi\sigma^2}$, then we can equivalently express the flatness factor as \cite{belfiore}
\begin{align}
\label{eqn:ff_theta} 
	\varepsilon_{\Lambda}(\sigma^2) = \gamma_{\Lambda}(\sigma^2)^{\frac{n}{2}}\Theta_{\Lambda}\left(e^{-\frac{1}{2\sigma^2}}\right)-1.
\end{align}

From the definition of the flatness factor, it is clear that a small flatness factor implies a more uniform distribution.

\section{Theta Series and the Compute-and-Forward Relaying Strategy}
\label{sec:caf}

In this section, we consider a protocol known as \emph{compute-and-forward} relaying \cite{nazer}. This protocol was proposed to harness the interference in an advantageous way. Namely, in wireless communications, a single transmission is heard by all near-enough receivers. Similarly, a receiver will hear all signals transmitted in the vicinity, not only the signals intended to them. This is referred to as interference, which degrades the reception quality. Several protocols have been proposed in the literature to remedy this degradation. The most prominent ones are \emph{decode-and-forward, compress-and-forward}, and \emph{amplify-and-forward}. For more details on these protocols, we refer to \cite{nazer} and references therein. The compute-and-forward strategy simultaneously aims at protection against noise and exploitation of  interference  for  cooperative gains. In contrast to compress-and-forward and amplify-and-forward, which can be seen as converting a network into a set of \emph{noisy} linear equations, the compute-and-forward converts it in to a set of \emph{reliable} linear equations. The compute-and-forward protocol has been shown to be superior at moderate signal quality levels, where both noise and interference play a non-negligible role.

Analyzing the \emph{maximum-likelihood} (ML) metric in the compute-and-forward context, we show how the flatness factor of a certain lattice enters the picture \cite{belfiore}, and relate this random lattice to the code lattice at the transmitter. We then utilize the derived theta series approximation to analyze the performance of various lattices with respect to an explicit design criterion. Namely, we show that in order to maximize the flatness factor of the random lattice, it suffices to maximize that of the code lattice.%v

In this article we will only consider real valued channels, which are also studied in the original article \cite{nazer} and additionally assumed in \cite{belfiore, belfiore2}. We refer to \cite{nazer} for the complex alternative\footnote{As shown in \cite{nazer}, a complex channel output can be treated as two separate real equations.}. Assume that $K > 1$ transmitters want to communicate to a single destination, aided by $M$ intermediate relays which, operating under the original compute-and-forward strategy attempt to decode an integer linear combination of the transmitted messages. We assume that each user, relay, and destination is equipped with one antenna only. The model is depicted in Figure~\ref{fig:caf_system_total}.
\begin{figure}[!h]
\centering
\begin{tikzpicture}
	\node[relay] (r1) {Relay 1};
	\node[below=0.2 of r1] (vert2) {$\vdots$};
	\node[relay,below=0.3 of vert2] (rm) {Relay M};
	\node[left=4 of vert2] (vert1) {$\vdots$};
	\node[user, above=0.2 of vert1] (t1) {T$_1$}
		edge[pil] node[pos=0.3,above,xshift=0.1cm]{\scriptsize $h_{11}$} (r1.west)
		edge[pil] node[pos=0.3,above,xshift=0.1cm]{\scriptsize $h_{M1}$} (rm.west);
	\node[user, below=0.3 of vert1] (tl) {T$_K$}
		edge[pil] node[pos=0.3,below,xshift=0.1cm]{\scriptsize $h_{1K}$} (r1.west)
		edge[pil] node[pos=0.3,below,xshift=0.1cm]{\scriptsize $h_{MK}$} (rm.west);
	\node[user, right=3 of vert2] (dest) {Dest.}
		edge[pil_rev] (r1.east)
		edge[pil_rev] (rm.east);
	\node[above=0.4 of t1, xshift=2cm] (first) {First Hop};
	\node[right=2.5 of first] (second) {Second Hop};
\end{tikzpicture}
\caption{System model with $K > 1$ transmitters and $M > K$ relays connected to a destination.}
\label{fig:caf_system_total}
\end{figure}

The first hop from the transmitters to the relays is modeled as a Gaussian fading channel, while it is usually assumed that the relays are connected to a destination with error-free bit pipes with unlimited capacities. We will henceforth focus on the first hop. 

The sources want to communicate messages $\tb{w}_k \in \F_p^s$ to the destination, which are encoded into $n$-dimensional codewords $\mb{x}_k \in \Lambda_{k,F} \subset \R^n$ before transmission. Here, $\Lambda_{k,F}$ is a full rank lattice employed by transmitter $k$, acting as the fine lattice in the nested code $\mc{C}_k(\Lambda_C,\Lambda_{k,F}) = \left\{\left[\mb{x}\right] \in \Lambda_{k,F} \left(\bmod\ \Lambda_C \right) \mid \mb{x} \in \Lambda_{k,F} \right\}$. We impose the usual symmetric power constraint $\frac{1}{n} E\left[||\mb{x}_k||^2\right] \leq P$ for all $k$. We can interpret the coarse lattice $\Lambda_C$ as the structure imposing the power constraint on the codewords, which allows us to ignore the specific definition of $\Lambda_C$ in the remainder of this section. The observed signal at relay $m$ can be expressed as
\begin{align}
	\mb{y}_m = \sum_{k=1}^K h_{mk}\mb{x}_k + \mb{n}_m,
\end{align}
where $\mb{n}_m$ is additive white Gaussian noise with variance $\sigma^2$, and the channel coefficients $h_{mk}$ are i.i.d. with normalized unit variance $\sigma_h^2 = 1$. Here, the \emph{signal-to-noise ratio} ($\snr$) is $\rho = P/\sigma^2$. The compute-and-forward strategy involves transforming the above random linear combination to a deterministic one and treating the rest of the equation as noise. We will describe this process next, leading to Eq. \eqref{eqn:caf_channel}. 

Channel state information is only available at the relays; more specifically, each relay only knows the channel $\mb{h}_m^t = (h_{m1}, \ldots, h_{mK})$ to itself. Operating under the original compute-and-forward protocol, a fixed relay selects a scalar $\alpha_m \in \R$, as well as an integer vector $\mb{a}_m^t = (a_{m1}, \ldots, a_{mK})$, and attempts to decode a linear combination of the received codewords with coefficients $a_{mk}$. For $\tilde{\mb{y}}_m := \alpha_m\mb{y}_m$, the channel output is modified to read
\begin{align}
\label{eqn:caf_channel}
	\tilde{\mb{y}}_m = \sum\limits_{k=1}^K {a_{mk}\mb{x}_k}  + \sum\limits_{k=1}^K{(\alpha_m h_{mk}-a_{mk})\mb{x}_k} + \alpha_m \mb{n}_m.
\end{align}
The so-called \emph{effective noise} 
\begin{align}
	\mb{n}_{\mathrm{eff}} := \sum_{k=1}^K{(\alpha_m h_{mk}-a_{mk})\mb{x}_k} + \alpha_m \mb{n}_m
\end{align} 
is no longer Gaussian. 

Upon observing the faded superposition of transmitted codewords, each relay proceeds in the same fashion in order to decode a linear combination. We can hence focus on a single relay and, for ease of notation, drop the subscript $m$ henceforth. The focused system model, now resembling a $K$-user multiple-access channel, is illustrated in Figure~\ref{fig:caf_system_onerelay}.
\begin{figure}[!h]
\centering
\begin{tikzpicture}
	\node[relay] (r) {Relay};
	\node[left=3 of r] (vert1) {$\vdots$};
	\node[user, above=0.2 of vert1] (t1) {T$_1$}
		edge[pil] node[pos=0.3,above]{\scriptsize $h_{1}$} (r.west);
	\node[user, below=0.3 of vert1] (tl) {T$_K$}
		edge[pil] node[pos=0.3,below]{\scriptsize $h_{K}$} (r.west);
	\node[right=.3 of r] (lincomb) {$\leadsto\ \mb{y} = \sum\limits_{k=1}^{K}{h_k \mb{x}_k} + \mb{n}$};
	\node[left=1 of t1] (x1) {$\underset{\in \F_q^k}{\omega_1} \overset{\mc{E}_1}{\longmapsto} \underset{\in \Lambda_{1,F}/\Lambda_C}{\mb{x}_1}$}
		edge[pil] (t1.west);
	\node[left=1 of tl] (xl) {$\underset{\in \F_q^k}{\omega_K} \overset{\mc{E}_K}{\longmapsto} \underset{\in \Lambda_{K,F}/\Lambda_C}{\mb{x}_K}$}
		edge[pil] (tl.west);
	\node[above=0.05 of r] (z) {\quad\ ${\Big\downarrow} \oplus \mb{n}$};
\end{tikzpicture}
\caption{System model focused on the first hop, with $K > 1$ transmitters and a fixed relay.}
\label{fig:caf_system_onerelay}
\end{figure}
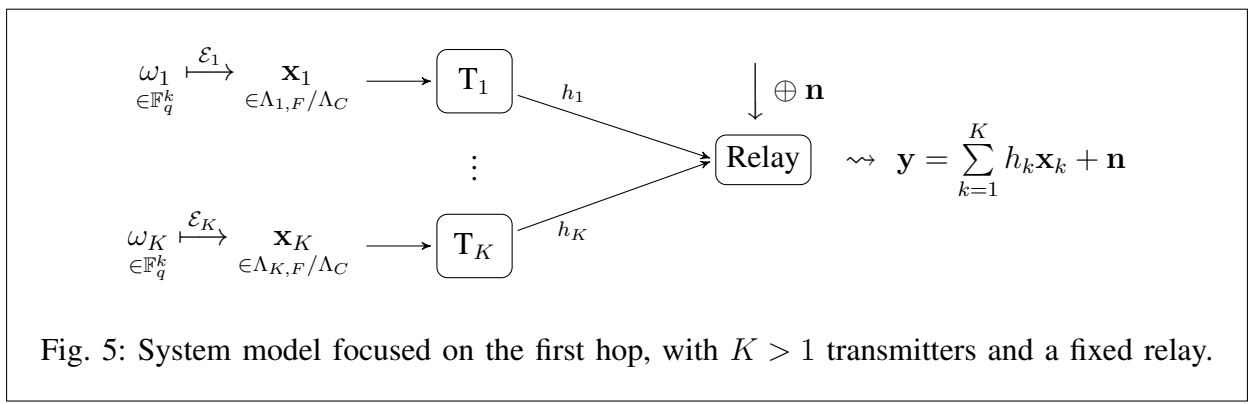

An important performance metric of the compute-and-forward protocol is the so-called \emph{computation rate}. If $\mc{R}_{M}(k) = \frac{s}{n}\log p$ denotes the \emph{message rate} at transmitter $k$, then the relay is able to decode a linear combination involving the codewords whose corresponding message rates are smaller than the computation rate $\mc{R}_C(\mb{h},\mb{a})$ achieved by the relay, that is, which satisfy $\mc{R}_M \le \mc{R}_C$. The main results on the computation rate are shortly summarized below.

\begin{lemma}\cite{nazer,osmane}
	For a relay employing the original compute-and-forward strategy under a real-valued channel model, the computation rate region is maximized by choosing $\alpha$ as the minimum mean square error (MMSE) estimate
	\begin{align}
		\alpha_{\mathrm{MMSE}} = \frac{\rho\mb{h}^t \mb{a}}{1+\rho||\mb{h}||^2},
	\end{align}
resulting in the computation rate region
	\begin{align}
		\mc{R}_C(\mb{h},\mb{a}) = \frac{1}{2}\log^+\left(\left(||\mb{a}||^2-\frac{\rho(\mb{h}^t \mb{a})^2}{1+\rho||\mb{h}||^2}\right)^{-1}\right).
	\end{align}
Moreover, the optimal coefficient vector is the solution to the minimization problem 
	\begin{align}
	\label{eqn:svp}
		\mb{a}_{\mathrm{opt}} = \argmin\limits_{\mb{a}\in\Z^K\backslash\left\{\mb{0}\right\}} \mb{a}^t G \mb{a},
	\end{align}
	where $G = I_K - \frac{\rho\mb{h}\mb{h}^t}{1+\rho||\mb{h}||^2}$. Hence, $\mb{a}_{\mathrm{opt}}$ corresponds to the coefficient vector of the shortest vector in the lattice with Gram matrix $G$.
\end{lemma}

\begin{remark}
	The lattice shortest vector problem is in general a computationally hard problem. However, it has been shown recently that in certain instances in the context of compute-and-forward, \emph{e.g.}, for solving \eqref{eqn:svp}, it can be solved in polynomial time \cite{sahraei}. 
	
	A low-complexity approach assuming no cooperation between the relays has also been proposed in \cite{barreal_caf}.
\end{remark}

\subsection{Decoding Linear Equations}
\label{subsec:decoding}

For each $k$, let $\mc{C}_{k} := \mc{C}_k(\Lambda_C,\Lambda_{k,F})$ denote the nested lattice code employed by transmitter $k$. Assume that the fine lattices, possibly after reordering the indexes, are nested, $\Lambda_{1,F} \supseteq \Lambda_{2,F} \supseteq \cdots \supseteq \Lambda_{K,F}$. Since the codebook is finite for each transmitter, the codewords can be assumed to be equiprobable in $\mc{C}_k$. 

A relay attempts to decode $$\mb{y} = \sum_{k=1}^{K}{h_k \mb{x}_k}+\mb{n}$$ to a lattice point $$\left[\lambda\right] = \sum_{k=1}^K{a_k \mb{x}_k}\ (\bmod\ \Lambda_C)$$ in two steps:
\begin{itemize}
	\item[i)] Scale the received signal by a scalar $\alpha$, compute an equation coefficient vector $\mb{a}^t = (a_{1}, \ldots, a_{K})$ by solving \eqref{eqn:svp}, and decode an estimate $\hat{\lambda}$ of
	\begin{align}
		\lambda = \sum\limits_{k=1}^K{a_k \mb{x}_k} \in \Lambda_F := \sum\limits_{k=1}^K{a_k \Lambda_{k,F}}.
	\end{align}
	
	\item[ii)] Apply the modulo-lattice operation to shift the received signal back into $\mc{V}(\Lambda_C)$,
	\begin{align}
		\left[\lambda\right] = \lambda\ (\bmod\ \Lambda_C).
	\end{align}
\end{itemize} 

The requirement $\Lambda_{1,F} \supseteq \Lambda_{2,F} \cdots \supseteq \Lambda_{K,F}$ guarantees\footnote{Note that nesting is not necessary, but sufficient; more generally, it suffices to fix a common superlattice for all transmitters. We adopt the nested assumption to be consistent with \cite{nazer}.} that
	\begin{align}
		\Lambda_F = \sum\limits_{k=1}^K{a_k \Lambda_{k,F}}
	\end{align}
is a lattice. The crucial step is the first one, estimating $\hat{\lambda} \in \Lambda_F$. Originally, a nearest neighbor decoder is used for this estimation. As this method is only optimal at high $\snr$, we employ ML decoding at the relay instead.

\subsection{The ML Decoding Metric}
\label{subsec:ml}

Let $\Lambda_F = \sum_{k=1}^{K}{a_k \Lambda_k}$ be the lattice defined above. By the imposed norm constraint on the codewords, the desired lattice point $\lambda = \sum_{k=1}^{K}{a_k\mb{x}_k}$ is contained in a finite subset $L_F \subset\Lambda_F$, which is determined by the norm restriction of the original codewords as well as the coefficient vector $\mb{a}$. Thus, a relay can restrict its search space to $L_F$. We make this more precise in the following straightforward proposition. 

\begin{proposition}
\label{prop:lambda_norm}
	For a fixed coefficient vector $\mb{a}^t = (a_1,\ldots,a_K)$, the lattice point $\lambda$ is contained in the set
	\begin{align}
		L_F = \left\{ \lambda \in \Lambda_{k_{\mathrm{min}},F} \left| ||\lambda|| \le \sum\limits_{k=1}^{K}{|a_k|}\max\limits_{\mb{x} \in \mc{C}_{k_{\mathrm{min}}}}{\left\{||\mb{x}||\right\}}\right.\right\},
	\end{align}
	where $k_{\mathrm{min}} := \argmin\limits_{1 \le k \le K}\left\{a_k \neq 0\right\}$.
\end{proposition}

\begin{proof}
	By definition, $k_{\mathrm{min}}$ is the index of the first non-zero entry in the coefficient vector $\mb{a}$, hence, the index of the first codeword to be included in the targeted linear combination. 
	As $\Lambda_{1,F} \supseteq \cdots \supseteq \Lambda_{K,F}$, we have that $\mb{x}_{k} \in \mc{C}_{k_{\mathrm{min}}}$ for all $k \ge k_{\mathrm{min}}$. Consequently, each of the codewords involved in the linear combination satisfies $||\mb{x}_k|| \le \max\limits_{\mb{x} \in \mc{C}_{k_{\mathrm{min}}}}{\left\{||\mb{x}||\right\}}$. We conclude
\begin{align}
	||\lambda|| &= \left|\left|\sum\limits_{k=1}^{K}{a_k\mb{x}_k}\right|\right| \le \sum\limits_{k=1}^{K}{|a_k|||\mb{x}_k||} \\
	&\le \sum\limits_{k=1}^{K}{|a_k|}\max\limits_{\mb{x}\in \mc{C}_{k_{\mathrm{min}}}}{\left\{||\mb{x}||\right\}}.
\end{align} 
\end{proof}

In this context, ML decoding amounts to maximizing the conditional probability 
\begin{align}
	\hat{\lambda} &= \argmax\limits_{\lambda \in L_F}{\mathbb{P}\left[\alpha\mb{y}\mid\lambda\right]} \\
	&= \argmax\limits_{\lambda \in L_F}{\sum\limits_{\substack{(\mb{x}_i)_i \in (\mc{C}_i)_i \\ \sum\limits_{k=1}^K{a_k\mb{x}_k} = \lambda}}}{\mathbb{P}\left[\alpha\mb{y}|(\mb{x}_1,\ldots,\mb{x}_K)\right]\mathbb{P}\left[(\mb{x}_1,\ldots,\mb{x}_K)\right]}
\end{align}

The former factor in the above expression behaves as 
\begin{align}
	\mathbb{P}\left[\alpha\mb{y}|(\mb{x}_1,\ldots,\mb{x}_K)\right] \propto \exp\left\{-\frac{1}{2\sigma^2}\left|\left|\mb{y}-\sum\limits_{k=1}^K{h_k\mb{x}_k}\right|\right|^2\right\}.
\end{align}
Note that this is independent of $\alpha$. We define the function
\begin{align}
\label{eqn:flat_function}
		\varphi(\lambda) := \sum\limits_{\substack{(\mb{x}_i)_i \in (\mc{C}_i)_i \\ \sum\limits_{k=1}^K{a_k\mb{x}_k} = \lambda}}{\exp\left\{-\frac{1}{2\sigma^2}\left|\left|\mb{y}-\sum\limits_{k=1}^K{h_k\mb{x}_k}\right|\right|^2\right\}},
\end{align}
and using the assumption that the codewords are equiprobable in $(\mc{C}_1,\ldots,\mc{C}_K)$, we conclude that the estimate $\hat{\lambda}$ of $\lambda$ can be computed by solving
\begin{align}\label{eqn:basicmax}
	\hat{\lambda} = \argmax\limits_{\lambda \in L_F}{\varphi(\lambda)}.
\end{align}

\begin{remark}
\label{rem:no_decod_algo}
We are not proposing a decoding algorithm, but rather elucidating the behavior of the decoding metric and deriving a code design criterion. It has been shown in \cite{belfiore2} that in dimension $n=1$ decoding based on Diophantine approximation is optimal, and in the same article it was conjectured to be optimal for $n\ge 2$ as well. However, how to treat simultaneous Diophantine equations is a mathematically open problem, which would be needed for implementing the Diophantine decoder in higher dimensions. While other optimal decoding methods may be derived, related work, such as \cite{mejri} have to date only proposed efficient decoding algorithms in arbitrary dimensions for Gaussian channels.
\end{remark}

Our goal in the remainder of this section is to study the behavior of $\varphi(\lambda)$. To analyze the decoding metric, we first need to express the function $\varphi(\lambda)$ in terms of the lattice point $\lambda$. This is achieved in the following proposition, whose proof we include as important quantities will be defined within. We follow a similar procedure described in \cite{belfiore,belfiore2}, but in more generality.
\begin{figure*}[!t]
\centering
	\includegraphics[width=0.45\textwidth]{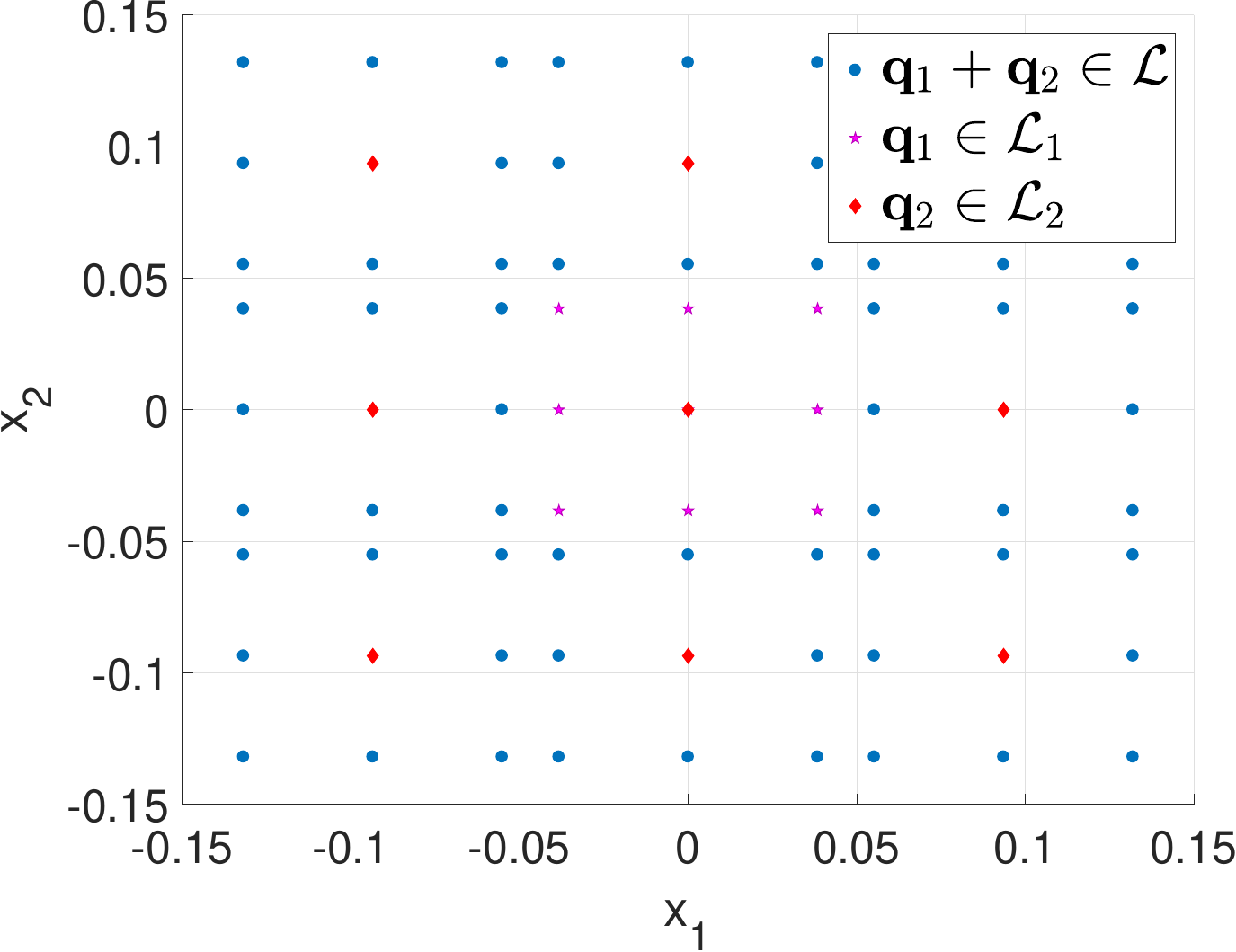}\hfill
	\includegraphics[width=0.45\textwidth]{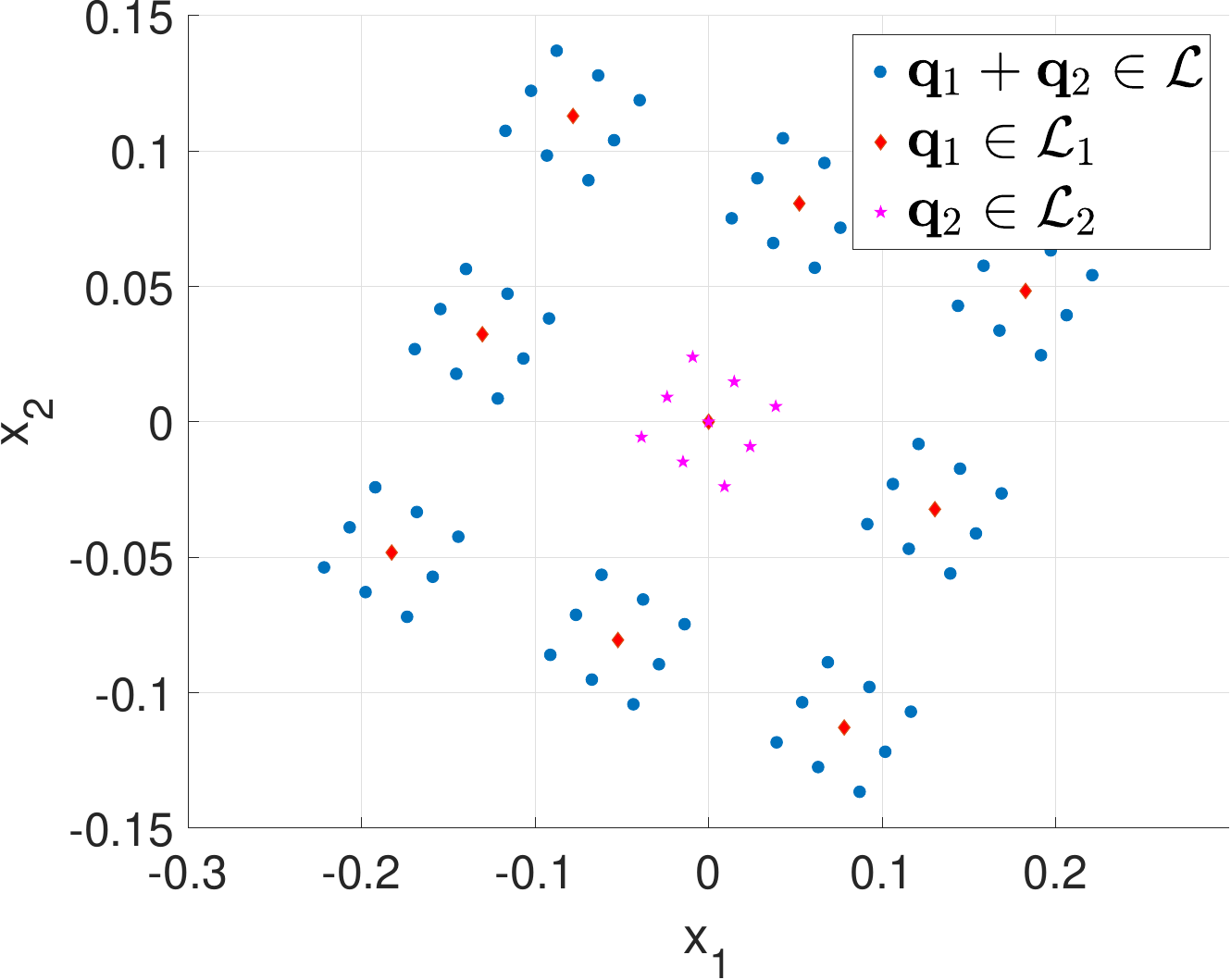}
	\caption{Sum of $(K-1) = 2$ lattices $\mc{L} = \mc{L}_1+\mc{L}_2$ in dimension $n = 2$. The depicted points correspond to coefficient vectors $\mb{z} \in [-p,p]^4$ with $p = 1$, and the density increases rapidly as $p$ grows. The employed code lattices are $\Z^2$ on the left, and $\Psi\left(\mathcal{O}_{\Q(\sqrt{5})}\right)$ on the right figure.}
	\label{fig:lattice_sum}
\end{figure*}
\begin{proposition}
\label{prop:decoding_manipulation}
	Let $\varphi(\lambda)$ be the decoding metric defined in \eqref{eqn:flat_function}. Then, $\varphi(\lambda)$ can be expressed in terms of the lattice point $\lambda$ as 
	\begin{align}
		\varphi(\lambda) = \sum\limits_{\mb{t} \in S \subset \Z^{nK}}{\exp\left\{-\frac{1}{2\sigma^2} \left|\left|\omega(\lambda) - M_{\mc{L}} \hat{U}\mb{t}\right|\right|^2\right\}},
	\end{align}
	where $S \subset \Z^{nK}$ is finite, $\omega(\lambda)$ is explicitly given in terms of $\lambda$, $\hat{U} \in \mat(n(K-1)\times nK,\R)$ and $M_{\mc{L}} \in \mat(n\times n(K-1),\R)$. 
\end{proposition}

\begin{proof}
For each transmitter $1 \le k \le K$, let $M_k \in \mat(n,\R)$ denote the generator matrix of $\Lambda_{k,F}$, and write $\mb{x}_k = M_k\mb{z}_k$ for some $\mb{z}_k \in \Z^n$. We define the matrix $M := \left[\begin{smallmatrix} a_1 M_1 & \cdots & a_K M_K \end{smallmatrix}\right] \in \mat(n\times nK,\R)$, where $\mb{a}^t = (a_1,\ldots,a_K)$ is the solution to \eqref{eqn:svp}, and express $\lambda$ as
\begin{align}
	\lambda &= \sum\limits_{k=1}^{K}{a_k \mb{x}_k} = \sum\limits_{k=1}^K{a_k M_k \mb{z_k}} \\
	&= \begin{bmatrix} a_1 M_1 & \cdots & a_K M_K \end{bmatrix}\begin{bmatrix} \mb{z}_1 \\ \vdots \\ \mb{z}_K\end{bmatrix} = M\mb{z}.
\end{align}

Let now\footnote{We will later choose a specific decomposition. However, any decomposition of this form suffices for decoding purposes.} $U \in \gl_{nK}(\R)$ be an invertible matrix such that
\begin{align}
\label{eqn:M_dec}
	\tilde{B} := MU = \begin{bmatrix} 0_{n\times n(K-1)} & \mid & B \end{bmatrix},
\end{align} 
where $B \in \mat(n,\R)$ is invertible. We proceed by decomposing the matrix $U$ into blocks $V_i \in \mat(n,\R)$ and $U_i \in \mat(n\times n(K-1),\R)$, as
\begin{align}
	U = \begin{bmatrix} U_1 & V_1 \\ \vdots & \vdots \\ U_K & V_K \end{bmatrix}.
\end{align} 

Let now $\tilde{\mb{r}} := U^{-1}\mb{z} = (\mb{r}^t,\mb{r}_n^t)^t$, where $\mb{r}_n$ denotes the last $n$ components of $\tilde{\mb{r}}$, and write
\begin{align}
	\lambda = M\mb{z} = \tilde{B}U^{-1}\mb{z} = \tilde{B}\tilde{\mb{r}} = B\mb{r}_n.
\end{align}
Note that $\mb{r}_n = B^{-1}\lambda$. To describe $\mb{r}$, the first $n(K-1)$ components of $\tilde{\mb{r}}$, let $\hat{U}$ be composed of the first $n(K-1)$ rows of $U^{-1}$. Then $\mb{r} = \hat{U}\mb{z}$. We can now write
\begin{align}	
	\begin{bmatrix} \mb{z}_1 \\ \vdots \\ \mb{z}_K \end{bmatrix} =\begin{bmatrix} U_1 & V_1 \\ \vdots & \vdots \\ U_K & V_K \end{bmatrix} \begin{bmatrix} \mb{r} \\ B^{-1}\lambda \end{bmatrix} = \begin{bmatrix}U_1 \mb{r} +  V_1 B^{-1}\lambda \\ \vdots \\ U_K \mb{r} + V_K B^{-1}\lambda \end{bmatrix},
\end{align}
and consequently rewrite the codewords $\mb{x}_k$ in terms of $\lambda$ as
\begin{align}
	\mb{x}_k &= M_k \mb{z}_k = M_k U_k \mb{r} + M_k V_k B^{-1}\lambda \\
	&= M_k U_k (\hat{U}\mb{z}) + \mu_k(\lambda),
\end{align}
where $\mu_k(\lambda) := M_k V_k B^{-1}\lambda$. For $\nu_k := M_k U_k (\hat{U}\mb{z})$, the exponent of $\varphi(\lambda)$ now takes the form

\begin{align}
	&\left|\left|\mb{y} - \sum\limits_{k=1}^K{h_k \mb{x}_k}\right|\right|^2 = \left|\left|\mb{y} - \sum\limits_{k=1}^K{h_k(\mu_k(\lambda) +\nu_k)}\right|\right|^2 \\
	&= \left|\left|\left(\mb{y} - \sum\limits_{k=1}^K{h_k \mu_k(\lambda)}\right)-\sum\limits_{k=1}^K{h_k \nu_k}\right|\right|^2.
\end{align}

To further simplify the expression, define the matrix
\begin{align}
\label{eqn:L_gen_matrix}
	M_{\mc{L}} := \sum\limits_{k=1}^K{h_k M_k U_k},
\end{align}
which allows us to rewrite $\varphi(\lambda)$ explicitly in terms of $\lambda$ as 
\begin{align}
\label{eqn:ml_decoding}
	\varphi(\lambda) = \sum\limits_{\mb{t} \in S \subset \Z^{nK}}{\exp\left\{-\frac{1}{2\sigma^2} \left|\left| \omega(\lambda) - M_{\mc{L}} \hat{U}\mb{t}\right| \right|^2\right\}}.
\end{align}
Here $S \subset \Z^{nK}$ is finite and we have defined 
\begin{align}
	\omega(\lambda) := \mb{y} - \sum\limits_{k=1}^K{h_k \mu_k(\lambda)}.
\end{align} 
\end{proof}

We state a lemma related to the structure defined by the matrix $M_{\mc{L}}$ for future reference, and quickly discuss the consequences. 
\begin{lemma}
\label{lem:sum_of_lattices}
	Let $M_{\mc{L}} = \sum\limits_{k=1}^{K}{h_k M_k U_k}$ be the matrix defined in \eqref{eqn:L_gen_matrix}. Then $M_{\mc{L}}$ defines a subgroup $\mc{L}$ of $\R^n$ of rank $n(K-1)$, which can only be discrete for $K = 2$. Hence, for $K \ge 3$, $\mc{L}$ is not a lattice almost surely, \emph{i.e.}, with probability one.
\end{lemma}

\begin{remark}
\label{rmk:lattice_sum}
	We remark that the authors in \cite{belfiore,belfiore2} are not aiming at analyzing the behavior of $\varphi(\lambda)$ for actual resulting lattice sums $\mc{L}$. The structure of $\mc{L}$ has only been studied in the case $K = 2$, and consequently, $\mc{L}$ has been commonly believed to be a lattice for any number of transmitters. By Lemma~\ref{lem:sum_of_lattices}, $\mc{L}$ is a lattice for $K = 2$, but lacks a discrete structure when $K > 2$. The main problem is the effect of the random channel coefficients $h_k$ and, as an important implication, the function $\varphi(\lambda)$ does not converge if the sum ranges over all of $\mc{L}$. This fact has dramatic consequences, as it implies that the tools developed in \cite{belfiore} for analyzing the behavior of $\mc{L}$ can only be applied in the case $K = 2$. 
\end{remark}

In general $\mc{L} = \sum_{i=1}^{K-1}{\mc{L}_i}$ is a sum of $(K-1)$ lattices, \emph{i.e.}, consists of vectors of the form $\mb{q} = \sum_{i=1}^{K-1}{\mb{q}_i}$, where $\mb{q}_i \in \mc{L}_i$. An example of a finite subset $\overline{\mc{L}} \subset \mc{L}$ for a sum of 2 lattices $\mc{L} = \mc{L}_1 + \mc{L}_2$ ($K = 3$) for $n = 2$ and a fixed channel vector is depicted in Figure~\ref{fig:lattice_sum} for illustrative purposes.

In the proof of Proposition~\ref{prop:decoding_manipulation}, we assumed the existence of a matrix $U \in \gl_{nK}(\R)$ which yields the desired decomposition \eqref{eqn:M_dec}. For a general matrix $U \in \gl_{nK}(\R)$, its inverse is a matrix with coefficients in $\R$, and hence $\mb{r} = \hat{U}\mb{z}$ is not an integer vector. Thus, $M_{\mc{L}}\hat{U}\mb{z}$ cannot be interpreted as an element of the lattice sum $\mc{L}$.

In \cite{belfiore, belfiore2, mejri}, the authors propose a decomposition based on the \emph{Hermite normal form} (HNF) of $M$. While the use of this specific decomposition has certain disadvantages, for example it only allows to consider integer lattices at the transmitter, it also allows to further simplify the decoding expression. Using the HNF, the matrix $U$ is unimodular, \emph{i.e.}, $U \in \gl_{nK}(\Z)$. In this special situation, we have that $\hat{U} \in \mat(n(K-1)\times nK,\Z)$, and consequently $\mb{r} = \hat{U}\mb{z} \in \Z^{n(K-1)}$. This allows to further simplify the ML decoding decision \eqref{eqn:ml_decoding} to read
\begin{align}
\label{eqn:ml_decoding_hnf}
	\hat{\lambda} = \argmax\limits_{\lambda\in L_F}{\sum\limits_{\mb{q} \in \overline{\mathcal{L}}}\exp\left\{\frac{1}{2\sigma^2}\left|\left|\omega(\lambda)-\mb{q}\right|\right|^2\right\}},
\end{align}
where $\mb{q} = M_{\mc{L}}\mb{z}$ with $\mb{z} \in \Z^{n(K-1)}$ ranges over a finite subset $\overline{\mc{L}} \subset \mc{L}$.

Nonetheless, any decomposition yielding a matrix in the form \eqref{eqn:M_dec} allows for ML decoding at the relay.

\subsection{The behavior of $\varphi(\lambda)$}
\label{subsec:ml_behavior}

We move on to analyze the behavior of the function 
\begin{align}
	\varphi(\lambda) = \sum\limits_{\mb{t} \in S \subset \Z^{nK}}{\exp\left\{\frac{1}{2\sigma^2}\left|\left|\omega(\lambda) - M_{\mc{L}}\hat{U}\mb{t} \right|\right|^2\right\}},
\end{align} 
which, as indicated in \cite{belfiore}, can be \emph{flat} for certain parameters leading to ambiguous decoding decisions and ultimately resulting in decoding errors. We begin by illustrating the behavior of $\varphi(\lambda)$ in Figure~\ref{fig:flatness} for dimensions $n = 1$ and $2$. In order to show that the flatness behavior of $\varphi(\lambda)$ prevails when using a decomposition other than the HNF, as well as when employing non-integer lattices, we use the $LQ$-decomposition of $M$. Here $M = LQ$, where $L$ is lower triangular and $Q$ unitary, and we choose $U := Q$, cf. \eqref{eqn:M_dec}.
\begin{figure*}[!t]
\centering
	\includegraphics[width=0.23\textwidth]{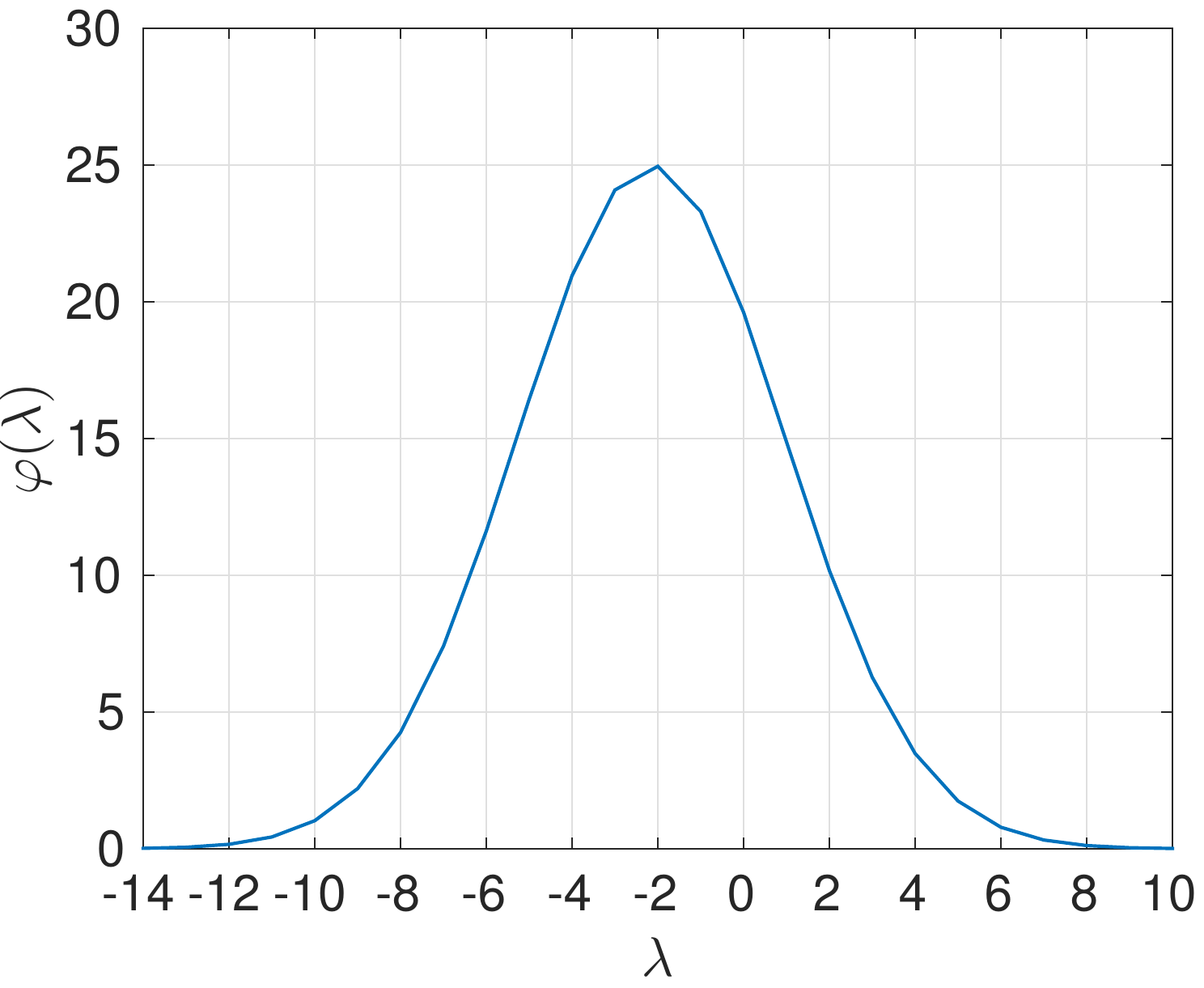} \hfill
	\includegraphics[width=0.23\textwidth]{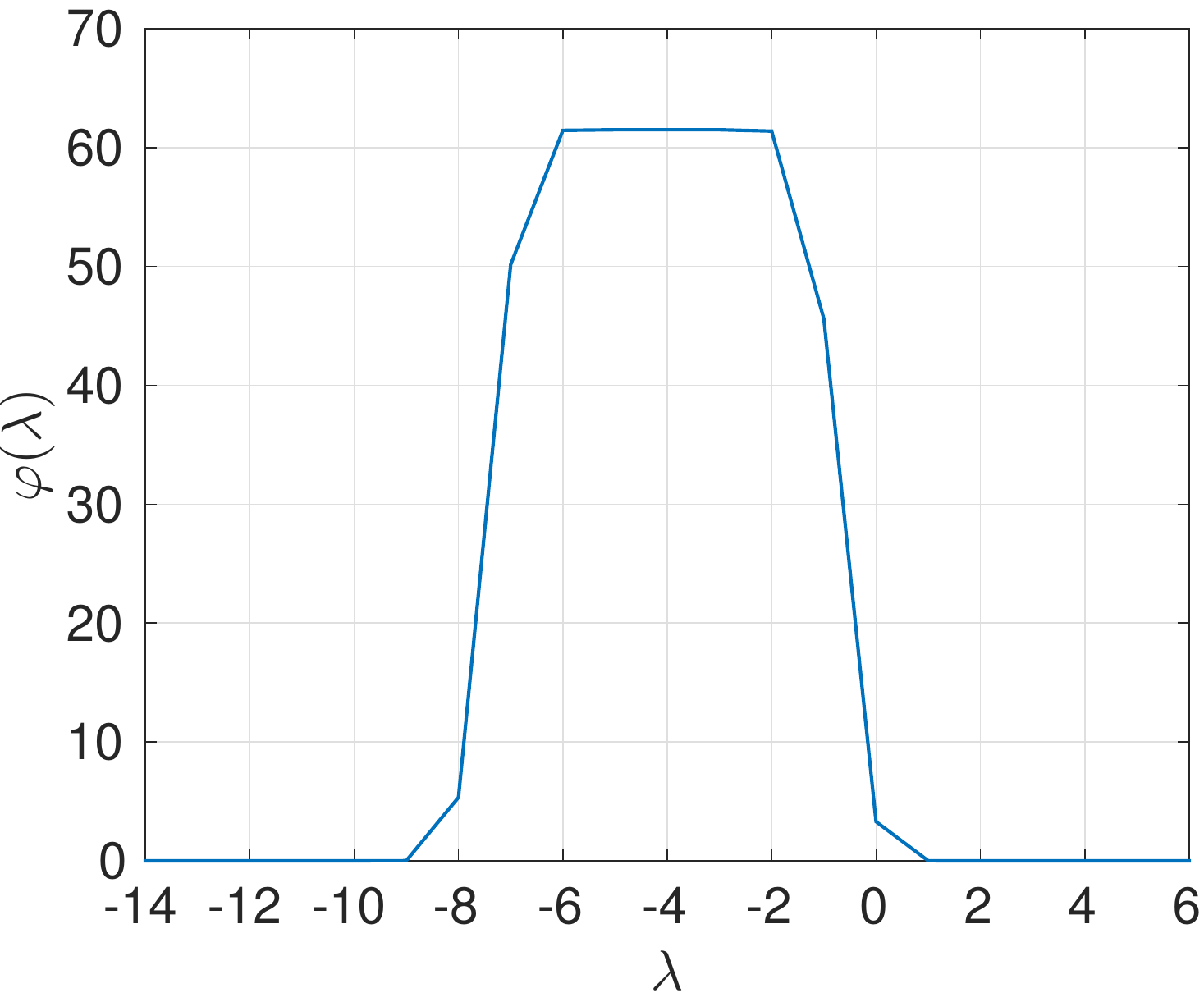} \hfill
	\includegraphics[width=0.23\textwidth]{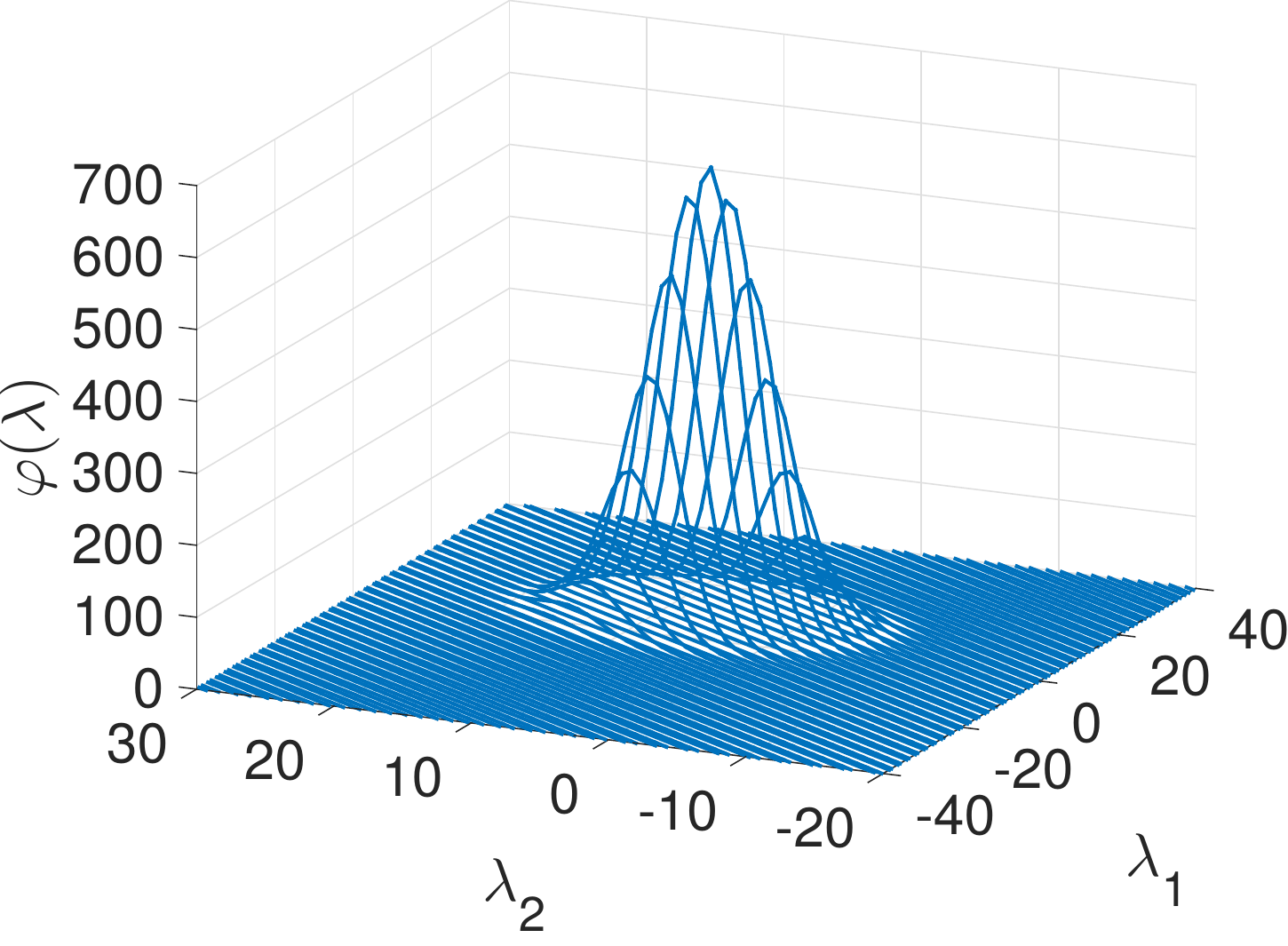} \hfill
	\includegraphics[width=0.23\textwidth]{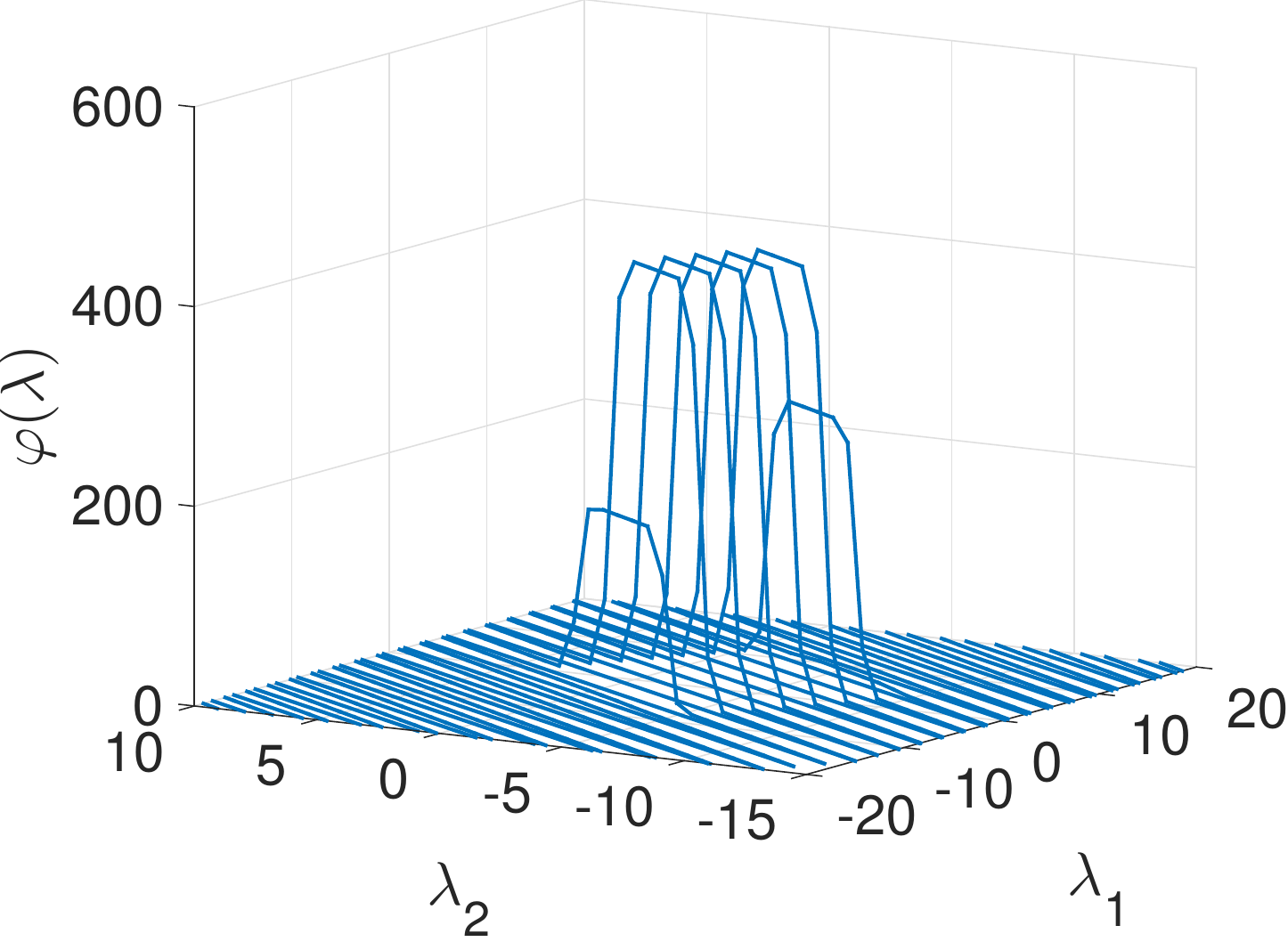}
	\caption{Behavior of $\varphi(\lambda)$ for $K = 2$ transmitters in dimension $n = 1$ (left) with $\Lambda = \Z$, and $n = 2$ (right) with $\Lambda = \Psi\left(\mathcal{O}_{\Q(\sqrt{5})}\right)$.}
	\label{fig:flatness}
\end{figure*}

In order to decode the lattice point $\lambda$, the relay needs to solve the maximization problem \eqref{eqn:basicmax}. We adopt two necessary restrictions.

\begin{itemize}
	\item[i)] The definition of the flatness factor involves the volume of the considered lattice. Hence, the analysis of $\varphi(\lambda)$ in terms of this quantity only makes sense when the volume $\vol{\mathcal{L}}$ is defined. We thus require that $\mc{L}$ is a lattice, \emph{i.e.}, $K = 2$ (cf. Lemma~\ref{lem:sum_of_lattices}).
	
	\item[ii)] Secondly, while any decomposition yielding the desired form allows the relay to solve \eqref{eqn:basicmax}, the matrix $\hat{U}$ may not be an integer matrix. The fractional part $\mathrm{frac}(\hat{U}\mb{t}) = \hat{U}\mb{t} - \mathrm{int}(\hat{U}\mb{t})$ may complicate the analysis of $\varphi(\lambda)$. To overcome this problem, we henceforth restrict to integer lattices, \emph{i.e.}, lattices with integer generator matrix. This allows us to choose the HNF as the employed decomposition, and consider the simplified expression \eqref{eqn:ml_decoding_hnf}.
\end{itemize}

\begin{remark}
	An extension to the case $K > 2$ seems necessary, as numerical results suggest that the flat behavior prevails for more than two transmitters. A natural first step is to study the average flatness factor restricted to finite sets of the lattices constituting $\mc{L}$, as a straightforward generalization of the flatness factor for a sum of lattices $\mc{L}$ is not obvious. This was considered in a preliminary version of this article. However, the relevance of such an approach needs to be verified, and numerical simulations are currently too expensive. 
\end{remark}

If the intermediate relay aims to decode a linear combination of $K = 2$ codewords, the ML decoding metric \eqref{eqn:ml_decoding} is a sum over lattice points, as repeatedly remarked previously. This allows us to characterize the behavior of $\varphi(\lambda)$ in terms of the flatness factor of the lattice $\mc{L}$ (cf. \eqref{eqn:ff_theta}). 

\begin{definition}
\label{def:ff_ml}
Let $K = 2$. The \emph{flatness factor} of $\varphi(\lambda)$ is defined as the flatness factor of $\mc{L}$,
\begin{align}
	\varepsilon_{\varphi(\lambda)}(\sigma^2) := \varepsilon_{\mc{L}}(\sigma^2).
\end{align}
\end{definition}

\begin{remark}
	The description of $\varepsilon_\Lambda(\sigma^2)$ in \eqref{eqn:ff_theta} allows to study the flatness factor as a function of the noise variance $\sigma^2$. In the context of compute-and-forward, we need $\varepsilon_{\varphi(\lambda)}(\sigma^2)$ to be as large as possible, as by the definition large values imply a distinctive maximum, which inhibits a flat behavior of the related function $\varphi(\lambda)$.
\end{remark}

Initially, studying the lattice flatness factor $\varepsilon_{\varphi(\lambda)}(\sigma^2)$ boils down to studying the flatness factor of the random lattice $\mc{L}$ which results at the relays. In order to have a reliable performance in the considered setting, we should choose lattices at the transmitter which are good for reliable communications, \emph{i.e.}, protect against noise and fading, while maximizing the flatness factor of the resulting lattice $\mc{L}$. By adopting the two restrictions listed above, it turns out that $\mc{L}$ can be related to the lattices employed at the transmitter, a link which we make explicit in Theorem~\ref{thm:lattice_equivalent} below. The consequences of the theorem are that maximizing the flatness factor of $\mc{L}$ amounts to maximizing the flatness factor of the original lattice. 

\begin{theorem}
\label{thm:lattice_equivalent}
	Let $K = 2$, and let $\Lambda_1,\Lambda_2 \subset \R^n$ be full integer lattices such that if $M_{\Lambda}$ is the generator matrix of $\Lambda_1$, then there exists $c \in \Z\backslash\left\{0\right\}$ such that $c M_{\Lambda}$ is the generator matrix for $\Lambda_2$. Hence, $\Lambda_1 \supseteq \Lambda_2$ are nested. Then, employing the Hermite normal form decomposition, the lattices $\mc{L}$ and $\Lambda_1$ are equivalent. 
\end{theorem}

\begin{proof}
	We determine the generator matrix $M_{\mc{L}}$ of the lattice $\mc{L}$. Assume that $\mb{a}^t = (a_1,a_2)$ is the coefficient vector determining the linear combination to be decoded. As $\mb{a}$ is the solution to a shortest vector problem, we have $\gcd(a_1,a_2) = 1$. Define the matrix 
	\begin{align}
		M := \begin{bmatrix} a_1 M_{\Lambda} & a_2 c M_{\Lambda} \end{bmatrix}.
	\end{align}

	Since we have $\mb{a}^t \neq (0,0)$, the matrix $M$ has full-rank. Hence, there always exist $U \in \gl_{2n}(\Z)$ and $B \in \mat(n,\Z)$ invertible, such that $MU = \begin{bmatrix} 0_n & B \end{bmatrix}$ is in HNF. If we write $A_1 := \diag\left\{a_1\right\}_{i=1}^{n}$, $A_2 := \diag\left\{c a_2\right\}_{i=1}^{n}$, and decompose the matrix $U$ into $n\times n$ blocks as 
	\begin{align}
		U = \begin{bmatrix} U_1 & V_1 \\ U_2 & V_2 \end{bmatrix},
	\end{align}
	we can write 
	\begin{align}
		MU &= \begin{bmatrix} a_1 M_{\Lambda} & a_2 c M_{\Lambda} \end{bmatrix} U \\
		&= \begin{bmatrix} M_{\Lambda} & M_{\Lambda} \end{bmatrix}\begin{bmatrix} A_1 & 0_n \\ 0_n & A_2 \end{bmatrix} \begin{bmatrix} U_1 & V_1 \\ U_2 & V_2 \end{bmatrix} \\
		&= \begin{bmatrix} M_{\Lambda} & M_{\Lambda} \end{bmatrix} \begin{bmatrix} A_1 U_1 & A_1 V_1 \\ A_2 U_2 & A_2 V_2 \end{bmatrix} = \begin{bmatrix} 0_n & B \end{bmatrix}.
	\end{align}	
	
	As $M_{\Lambda}$ generates a full lattice, it is invertible. We multiply by $M_{\Lambda}^{-1}$ from the left to get 
	\begin{align}
		&M_{\Lambda}^{-1}\begin{bmatrix} M_{\Lambda} & M_{\Lambda} \end{bmatrix} \begin{bmatrix} A_1 U_1 & A_1 V_1 \\ A_2 U_2 & A_2 V_2 \end{bmatrix}\\
		&= \begin{bmatrix} I_n & I_n \end{bmatrix}\begin{bmatrix} A_1 U_1 & A_1 V_1 \\ A_2 U_2 & A_2 V_2 \end{bmatrix} \\
		 &= \begin{bmatrix} (A_1 U_1 + A_2 U_2) & (A_1 V_1 + A_2 V_2) \end{bmatrix} \\
		 &= \begin{bmatrix} 0_n & M_{\Lambda}^{-1}B \end{bmatrix},
	\end{align}
	which yields the equations $A_1 U_1 + A_2 U_2 = 0_n$ and $A_1 V_1 + A_2 V_2 = M_{\Lambda}^{-1}B$. We can rewrite the first equation to read
	\begin{align}
		&\begin{bmatrix} A_1 & A_2 \end{bmatrix} \begin{bmatrix} U_1 \\ U_2 \end{bmatrix} \\
		&= \begin{bmatrix} a_1 & 0 & \cdots & 0 & c a_2 & 0 & \cdots & 0 \\ 0 & a_1 & & 0 & 0 & c a_2 & & 0 \\ \vdots & & \ddots & \vdots & \vdots & & \ddots & \vdots \\ 0 & \cdots & & a_1 & 0 & & \cdots & c a_2 \end{bmatrix} \begin{bmatrix} U_1 \\ U_2 \end{bmatrix} = 0_n.
	\end{align}

	This equation is satisfied if and only if 
	\begin{align}
		\colspan\left(\begin{bmatrix} U_1 \\ U_2 \end{bmatrix}\right) &\subseteq \ker\left(\begin{bmatrix} A_1 & A_2 \end{bmatrix}\right) \\
		&= \vecspan\left\{\begin{bmatrix} ca_2 \cdot e_i \\ -a_1 \cdot e_i \end{bmatrix}\right\}_{i=1}^{n},
	\end{align}
	where $e_i$ is the $i^{\mathrm{th}}$ standard vector. In particular, we can always choose 
	\begin{align}
		U_1 = -\diag\left\{c a_2 \right\}_{i=1}^{n}; \quad U_2 = \diag\left\{a_1\right\}_{i=1}^{n}.
	\end{align}

	With this choice, the generator matrix of $\mc{L}$ simplifies to 
	\begin{align}
		M_{\mc{L}} &= h_1 M_{\Lambda} U_1 + h_2 c M_{\Lambda} U_2 \\
		&= -h_1 M_{\Lambda}\diag\left\{c a_2 \right\}_{i=1}^{n} + h_2 c M_{\Lambda}\diag\left\{a_1\right\}_{i=1}^{n} \\
		&= (h_2 c\diag\left\{a_1\right\}_{i=1}^{n} - h_1\diag\left\{c a_2 \right\}_{i=1}^{n})M_{\Lambda} = rM_{\Lambda}
	\end{align}
	for some $r \in \R$.
\end{proof}

This result motivates the study of the flatness factor of the code lattices. As these should be picked to be well conditioned for coding purposes, we only need to compute the flatness factor of reasonably conditioned lattices. Thus, the derived approximation $\Theta_{\Lambda}^{\mf{A}}(q)$ will suffice for that purpose. 
\begin{figure*}[!t]
\centering
	\includegraphics[width=.48\textwidth]{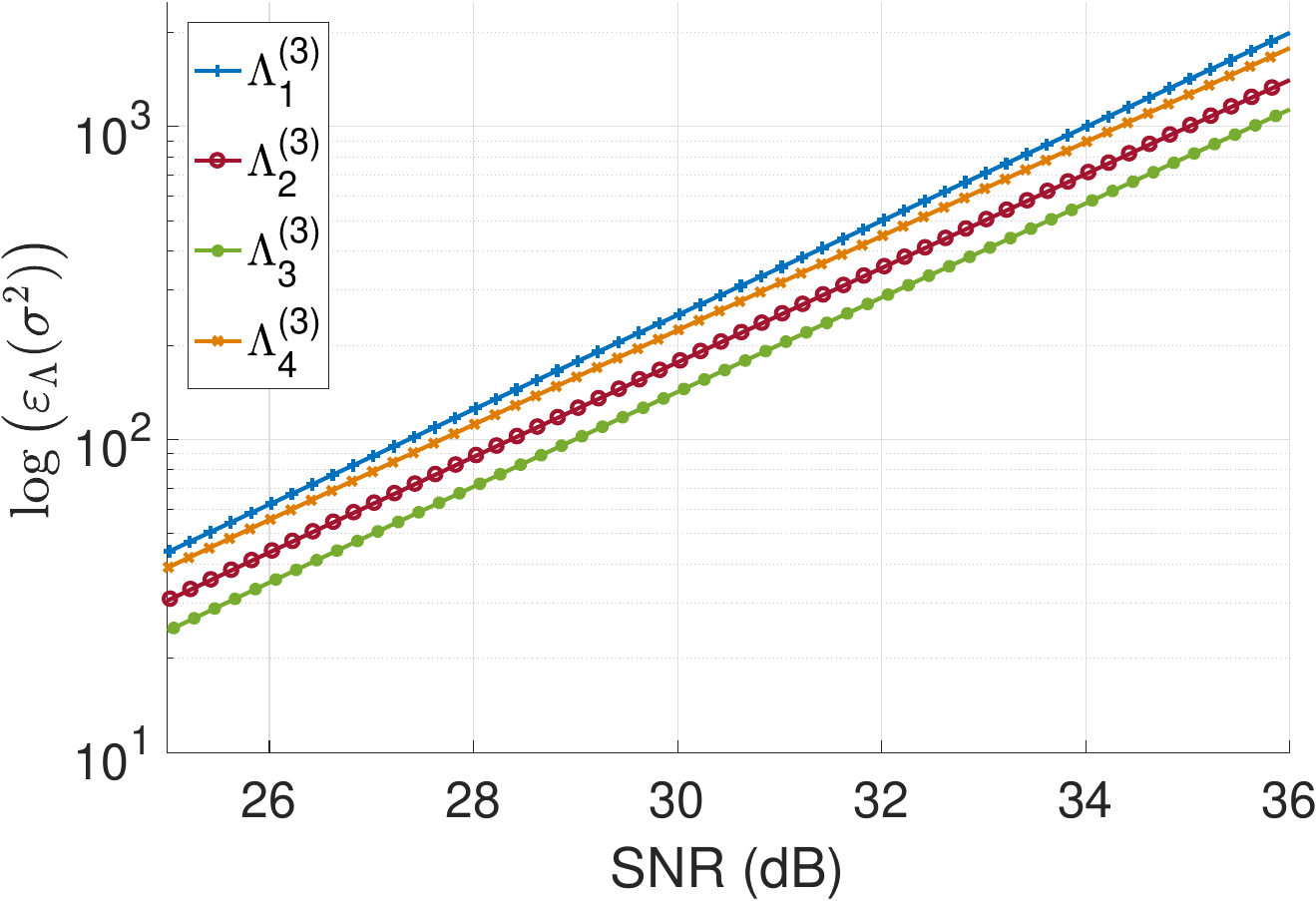} \hfill
	\includegraphics[width=.48\textwidth]{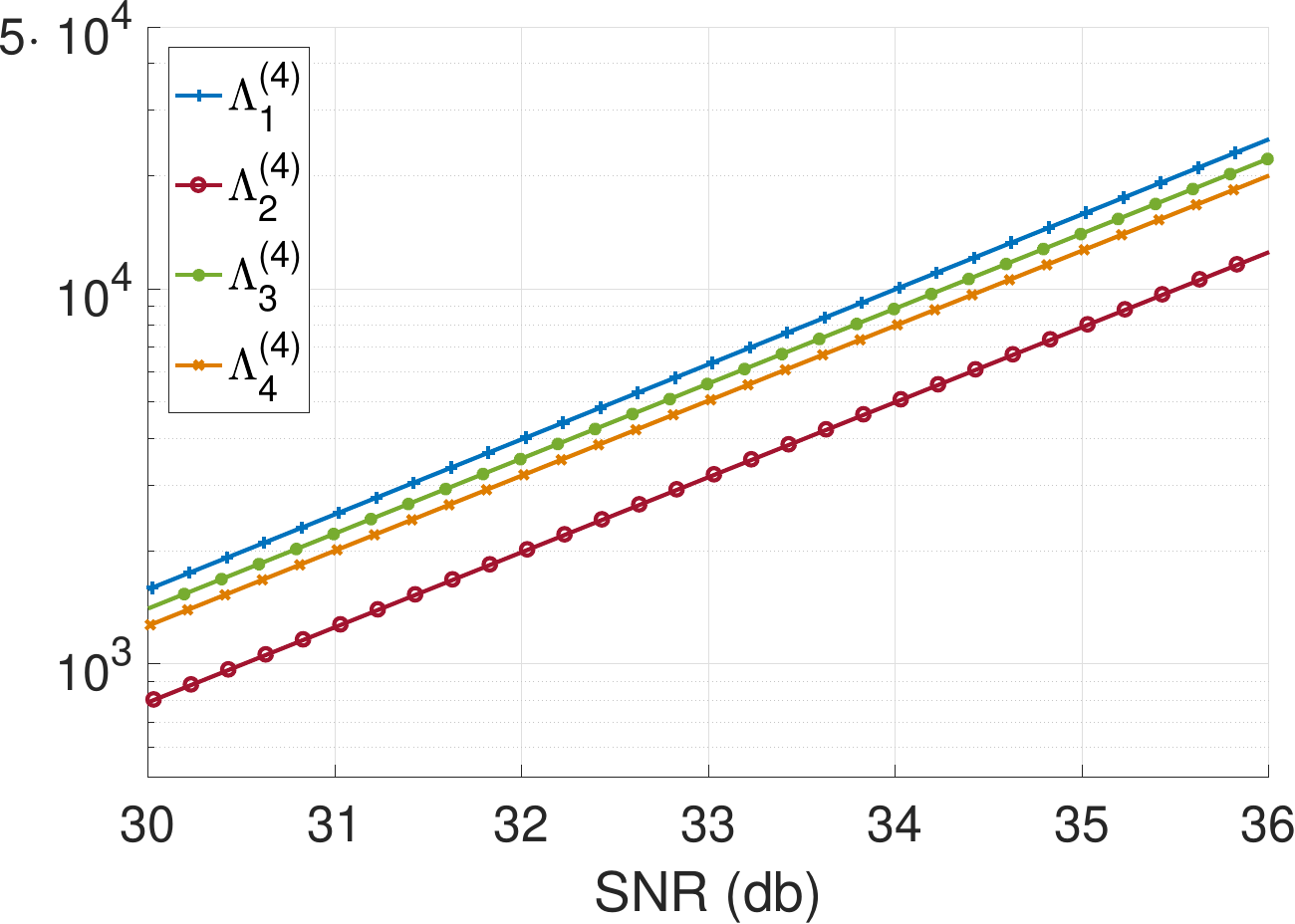}
	\caption{Flatness factors of $\Lambda$ for lattices of dimensions $n = 3$ (left) and $n = 4$ (right).}
	\label{fig:FF_lattices}	
\end{figure*} 
We consider the lattices $\Lambda_i^{(3)}$ and $\Lambda_i^{(4)}$, $1 \le i \le 4$, tabulated in Table~\ref{tab:wr_lattices} in Appendix~\ref{sec:app}. As the well-known lattices $\Z^n$, $D_n$ and the dual $D^\ast_n$ are all examples of well-rounded lattices, we consider additional lattices $\Lambda_4^{(3)}$, $\Lambda_3^{(4)}$ and $\Lambda_4^{(4)}$ which are well-rounded as well, for sake of consistency. These are found via computer search. 

\begin{remark}
Note that, as it should, the flatness factor of $\varphi(\lambda)$ is independent of the size of constellation, as it is simply the flatness factor of the unconstrained lattice $\mc{L}$. For a meaningful comparison, however, we fix a finite codebook for each of the considered lattices, and illustrate their flatness factor with respect to the power-dependent $\snr$, $\rho = P/\sigma^2$. 
\end{remark}

The average power $P$ for the employed constellation is also found in Table~\ref{tab:wr_lattices}. We compare the considered lattices in Figure~\ref{fig:FF_lattices}.  

Both in dimension $n = 3$ and $n = 4$, it is visible that the integer lattice $\Lambda_1^{(n)} = \Z^n$ performs best among the considered lattices with respect to the flatness factor criterion. This is in agreement with the observation in \cite{belfiore2} that the lattice $\mc{L}$ should not be dense. However, the density is not the only factor that plays a role, as visible from the plot in dimension $n = 3$. There, the best quantizer, $\Lambda_3^{(3)} = D_3^{\ast}$ exhibits the smallest flatness factor, even below the densest packing $\Lambda_2^{(3)} = D_3$. In dimension $n = 4$, the lattice $\Lambda_2^{(4)} = D_4$ is both the best quantizer and densest packing, and exhibits the smallest flatness factor. 

The quintessential statement, however, is not that the lattice $\Z^n$ is the one that should always be used. Indeed, the code lattice should firstly be chosen to perform well in compute-and-forward, and additionally exhibit a large flatness factor. This yields a potential trade-off in code design.

\section{Conclusions}
\label{sec:conclusions}

The main goal of this article was to derive a simple approximation of the theta series of a lattice. Our approximation can be shown to be a simple rational function.

We then studied maximum-likelihood decoding in the context of compute-and-forward relaying, and showed that partial code design criteria can be derived based on the so-called flatness factor of certain involved lattices. Using a particular matrix decomposition for manipulating the decoding metric, and adopting two important restrictions, we further prove that the code lattice at the transmitter and the random lattice at the relay are similar. This allows for a direct design criterion for the code lattice, rather than for the random lattice. Namely, the flatness factor of the code lattice should be maximized. 

As the flatness factor is directly related to the theta series of a lattice, it is hence crucial to be able to efficiently compute the latter quantity. Hence, for the purposes of empirically analyzing different lattices at the transmitter, the theta series approximation proves to be crucial, both in this context as in the context of wiretap coset code design, \emph{e.g.} the results obtained in \cite{barreal_wiretap}. 

This work allows to extend the framework in a variety of directions. First, as noted in this article, the decoding metric is only a sum over lattice points for $K = 2$ transmitters, and the analysis of its behavior becomes more complicated when $K \ge 3$, though numerical results show that the flatness behavior prevails. On the other hand, the used decomposition only allows for integer lattices and integer linear combinations. Following related work \cite{huang,tunali} where the linear combinations are allowed to be over the ring of integers of an algebraic number field, it would be of benefit to examine the decoding metric in this generalized setting. The Hermite normal form decomposition over the integers $\Z$ is only a special case, and the algorithm has been extended to arbitrary Dedekind domains. Thus using this generalized decomposition would allow to study algebraic lattices for code construction at the transmitters.

\section*{Acknowledgment} The authors would like to thank Dr. David Karpuk for his help towards this article. 


\newpage

\section*{Appendix}
\label{sec:app}

The table below serves as a summary of the characteristics of the lattices used for simulations, and introduces the employed notation. 
\begin{table*}[h!]
\centering
\resizebox{\columnwidth}{!}{%
\begin{tabular}{c|l||l|c|c|l||c|}
\hline
$\mathbf{n = 3}$ & Notation & $M_\Lambda$ & $\lambda_1$ & $\vol{\Lambda}$ & $\Theta_\Lambda(q)$ & P $\left(|\mc{C}|= 343\right)$\\
\hline \cline{2-7} 
 & $\Lambda_1^{(3)} = \Z^3$ & $I_3$ & 1 & 1 & $\theta_3^3(q)$ & 4 \\
 \cline{2-7}
 \rule{0pt}{.95cm}
 & $\Lambda_2^{(3)} = D_3 \cong A_3$ & $\begin{bmatrix} -1 & 1 & 0 \\ -1 & -1 & 1 \\ 0 & 0 & -1 \end{bmatrix}$ & 2 & 2 & $\frac{1}{2}(\theta_3^3(q)+\theta_4^3(q))$ & 8 \\[.45cm]
 \cline{2-7}
 \rule{0pt}{.95cm}
 & $\Lambda_3^{(3)} = D_3^\ast \cong A_3^\ast$ & $\begin{bmatrix} 2 & 0 & 1 \\ 0 & 2 & 1 \\ 0 & 0 & 1  \end{bmatrix}$ & 3 & 4 & $\theta_2(4q)^3+\theta_3(4q)^3$ & 16.6667 \\[.45cm]
 \cline{2-7}
 \rule{0pt}{.95cm}
 & $\Lambda_4^{(3)}$ & $\begin{bmatrix} 2 & 0 & 0 \\ 1 & -2 & 1 \\ 0 & -1 & -2  \end{bmatrix}$ & 5 & 10 & - & 20 \\[.45cm]
 \hline \cline{2-7}
 $\mathbf{n = 4}$ & \multicolumn{5}{c||}{} & P $\left(|\mc{C}|= 2401\right)$ \\
 \hline \cline{2-7}
 & $\Lambda_1^{(4)} = \Z^4$ & $I_4$ & 1 & 1 & $\theta_3^4(q)$ & 4 \\
 \cline{2-7}
 \rule{0pt}{1.3cm}
 & $\Lambda_2^{(4)} = D_4$ & $\begin{bmatrix} -1 & 1 & 0 & 0 \\ -1 & -1 & 1 & 0 \\ 0 & 0 & -1 & 1 \\ 0 & 0 & 0 & -1 \end{bmatrix}$ & 2 & 2 & $\frac{1}{2}(\theta_3^4(q)+\theta_4^4(q))$ & 8 \\[.9cm]
 \cline{2-7}
 \rule{0pt}{1.3cm}
 & $\Lambda_3^{(4)}$ & $\begin{bmatrix} 1 & 1 & -1 & 1 \\ 1 & -1 & 1 & 1 \\ -1 & 0 & 0 & 1 \\ 0 & 1 & 1 & 0 \end{bmatrix}$ & 3 & 8 & - & 12 \\[.9cm]
 \cline{2-7}
 \rule{0pt}{1.3cm}
 & $\Lambda_4^{(4)}$ & $\begin{bmatrix} -2 & 0 & 0 & 0 \\ 0 & 0 & 0 & -2 \\ 1 & 1 & -2 & 1 \\ 0 & 2 & 1 & 0 \end{bmatrix}$ & 5 & 20 & - & 20 \\[.9cm]
 \cline{2-7}
\end{tabular}}
\caption{Summary of the lattices employed for simulation results.}
\label{tab:wr_lattices}
\end{table*}

\end{document}